\documentclass{llncs} 

\usepackage{proof}
\usepackage{latexsym}
\usepackage{amssymb}

\newcommand\enc[2]{\{#1\}_{#2}}

\newcommand\spr[2]{\langle #1, #2 \rangle}
\newcommand\norm[1]{#1\!\!\downarrow\ }
\newcommand\normE[2]{#1\!\!\downarrow_{#2}\ }
\newcommand\pub[1]{\mathsf{pub}(#1)}
\newcommand\sign[2]{\mathsf{sign}(#1,#2)}
\newcommand\blind[2]{\mathsf{blind}(#1,#2)}

\def\pubsym{\mathsf{pub}}
\def\signsym{\mathsf{sign}}
\def\blindsym{\mathsf{blind}}
\def\nameset{\mathsf{N}}
\def\varset{\mathsf{V}}
\def\conset{\Sigma_C}

\def\Dscr{{\cal D}}

\def\Lscr{{\cal L}}
\def\Nscr{{\cal N}}
\def\Rscr{{\cal R}}
\def\Sscr{{\cal S}}

\begin{document}

\pagestyle{plain}

\title{A Proof Theoretic Analysis of Intruder Theories}
\author{Alwen Tiu and Rajeev Gor\'e}
\institute{
Logic and Computation Group \\ 
College of Computer Science and Engineering \\
The Australian National University 
}

\maketitle

\thispagestyle{empty}

\begin{abstract}
  We consider the problem of intruder deduction in security protocol
  analysis: that is, deciding whether a given message $M$ can be
  deduced from a set of messages $\Gamma$ under the theory of blind
  signatures and arbitrary convergent equational theories modulo
  associativity and commutativity (AC) of certain binary
  operators. The traditional formulations of intruder deduction are
  usually given in natural-deduction-like systems and proving
  decidability requires significant effort in showing that the rules
  are ``local'' in some sense. By using the well-known translation
  between natural deduction and sequent calculus, we recast the
  intruder deduction problem as proof search in sequent calculus, in
  which locality is immediate. Using standard proof theoretic methods,
  such as permutability of rules and cut elimination, we show that the
  intruder deduction problem can be reduced, in polynomial time, to
  the elementary deduction problems, which amounts to solving certain
  equations in the underlying individual equational theories.  We
  further show that this result extends to combinations of disjoint
  AC-convergent theories whereby the decidability of intruder
  deduction under the combined theory reduces to the decidability of
  elementary deduction in each constituent theory. Although various
  researchers have reported similar results for individual cases, our
  work shows that these results can be obtained using a systematic and
  uniform methodology based on the sequent calculus.

\end{abstract}

\section{Introduction}

One of the fundamental aspects of the analysis of security protocols is the model
of the intruder that seeks to compromise the protocols. In many situations, 
such a model can be described in terms of a deduction system which gives a formal
account of the ability of the intruder to analyse and synthesize messages. 
As shown in many previous works 
(see, e.g., ~\cite{Amadio00CONCUR,Boreale01ICALP,Comon-Lundh03LICS,Chevalier03LICS}), 
finding attacks on protocols can often be framed as
the problem of deciding whether a certain formal expression is derivable in
the deduction system which models the intruder capability. The latter is sometimes called
the {\em intruder deduction problem}, or the (ground) reachability problem.  
A basic deductive account of the intruder's capability 
is based on the so-called Dolev-Yao model, which assumes perfect encryption. 
While this model has been applied fruitfully to many situations, 
a stronger model of intruders is needed to discover certain types of attacks.
A recent survey \cite{Cortier06JCS} shows that attacks on several protocols used in real-world communication
networks can be found by exploiting algebraic properties of 
encryption functions. 

The types of attacks mentioned in \cite{Cortier06JCS} have motivated
many recent works in studying models of intruders in which the
algebraic properties of the operators used in the protocols are taken
into
account~\cite{Comon-Lundh03LICS,Chevalier03LICS,Abadi06TCS,Delaune06ICALP,Lafourcade07IC,Cortier07LPAR}.
In most of these, the intruder's capability is usually given as
a natural-deduction-like deductive system. As is common in natural
deduction, each constructor has a rule for introducing the constructor and
one for eliminating the constructor.
The elimination rule typically decomposes a term, reading
the rule top-down: {\em e.g.}, a typical elimination rule for a pair
$\spr M N$ of terms is:
$$
\infer[]
{\Gamma \vdash M}
{\Gamma \vdash \spr M N}
$$
Here, $\Gamma$ denotes a set of terms, which represents the terms
accumulated by the intruder over the course of its interaction with
participants in a protocol. While a natural deduction formulation of
deductive systems may seem ``natural'' and may reflect the meaning of
the (logical) operators, it does not immediately give us a proof
search strategy. Proof search means that we have to apply the rules
bottom up, and as the above elimination rule demonstrates, this
requires us to come up with a term $N$ which might seem arbitrary.
For a more complicated example, consider the following elimination
rule for {\em blind signatures} \cite{Fujioka92,KremerESOP05,Bernat06}.
$$
\infer[]
{\Gamma \vdash \sign M K}
{\Gamma \vdash \sign {\blind M R} K & \Gamma \vdash R}
$$
The basis for this rule is that the ``unblinding'' operation commutes
with signature.  Devising a proof search strategy in a natural
deduction system containing this type of rule does not seem trivial.
In most of the works mentioned above, in order to show the
decidability results for the natural deduction system, one needs to
prove that the system satisfies a notion of {\em locality}, i.e., in
searching for a proof for $\Gamma \vdash M$, one needs only to
consider expressions which are made of subterms from $\Gamma$ and $M.$
In addition, one has to also deal with the complication that arises
from the use of the algebraic properties of certain operators.

In this work, we recast the intruder deduction problem as proof search
in sequent calculus. A sequent calculus formulation of Dolev-Yao
intruders was previously used by the first author in a formulation of
open bisimulation for the spi-calculus~\cite{Tiu07APLAS} to prove
certain results related to open bisimulation.  The current work takes
this idea further to include richer theories.  Part of our motivation
is to apply standard techniques, which have been well developed in the
field of logic and proof theory, to the intruder deduction problem. In
proof theory, sequent calculus is commonly considered a better
calculus for studying proof search and decidability of logical
systems, in comparison to natural deduction. This is partly due to the
so-called ``subformula'' property (that is, the premise of every
inference rule is made up of subterms of the conclusion of the rule),
which in most cases entails the decidability of the deductive
system. It is therefore rather curious that none of the existing works
on intruder deduction so far uses sequent calculus to structure proof
search.  We consider the ground intruder deduction problem (i.e.,
there are no variables in terms) under the class of {\em AC-convergent
  theories}. These are equational theories that can be turned into
convergent rewrite systems, modulo associativity and commutativity of
certain binary operators. Many important theories for intruder
deduction fall into this category, e.g., theories for
exclusive-or~\cite{Comon-Lundh03LICS,Chevalier03LICS}, Abelian
groups~\cite{Comon-Lundh03LICS}, and more generally, certain classes
of monoidal theories~\cite{Cortier07LPAR}.

We show two main results.
Firstly, we show that the
decidability of intruder deduction under AC-convergent theories can be
reduced, in polynomial time, 
to {\em elementary intruder deduction problems}, which involve
only the equational theories under consideration.
Secondly, we show that the intruder deduction problem for a
combination of disjoint theories $E_1,\ldots,E_n$ can be reduced, in polynomial time, 
to the elementary deduction problem {\em for each theory $E_i$.} This
means that if the elementary deduction problem is decidable for each
$E_i$, then the intruder deduction problem under the combined theory
is also decidable.
We note that these decidability results are not really new,
although there are slight differences and improvements over the
existing works (see Section~\ref{sec:rel}).
Our contribution is more of a methodological nature. We
arrive at these results using rather standard proof theoretical
techniques, e.g., {\em cut-elimination} and permutability of inference
rules, in a uniform and systematic way. In particular, we obtain
locality of proof systems for intruder deduction, which is one of the
main ingredients to decidability results in 
\cite{Comon-Lundh03LICS,Chevalier03LICS,Delaune06ICALP,Delaune06IPL},
for a wide range of theories that cover those studied in these works. 
Note that these works deal with a more difficult problem of deducibility
constraints, which models {\em active intruders}, whereas we currently
deal only with passive intruders.  As future work, we plan to
extend our approach to deal with active intruders. 

The remainder of the paper is organised as follows. 
Section~\ref{sec:intruder} presents two systems for intruder theories,
one in natural deduction and the other in sequent calculus, and show that the two systems are equivalent.
In Section~\ref{sec:cutelim}, the sequent system is shown to enjoy cut-elimination. 
In Section~\ref{sec:normal}, we show that cut-free sequent derivations can be transformed into a certain
normal form. Using this result, we obtain another ``linear'' sequent system, from which the polynomial
reducibility result follows. Section~\ref{sec:examples} discusses several example theories 
which can be found in the literature. Section~\ref{sec:comb} shows that the sequent system 
in Section~\ref{sec:intruder} can be extended to cover any combination of disjoint AC-convergent theories,
and the same decidability results also hold for this extension. 
Detailed proofs can be found in the appendix.

\section{Intruder deduction under AC convergent theories}
\label{sec:intruder}

We consider the following problem of formalising, given a set of messages
$\Gamma$ and a message $M$, whether $M$ can be synthesized from the messages
in $\Gamma.$ We shall write this judgment as $\Gamma \vdash M.$ This is sometimes called
the `ground reachability' problem or the `intruder deduction' problem in the literature. 

Messages are formed from names, variables and function symbols.  We
shall assume the following sets: a countably infinite set $\nameset$
of names ranged over by $a$, $b$, $c$, $d$, $m$ and $n$; a countably
infinite set $\varset$ of variables ranged over by $x$, $y$ and $z$;
and a finite set $\conset = \{\pubsym, \signsym, \blindsym,
\spr{.}{.}, \enc{.}{.}\}$ of symbols representing the {\em
  constructors}. Thus $ \pubsym$ is a public key constructor,
$\signsym$ is a public key encryption function, $\blindsym$ is the
blinding encryption function (as in
\cite{Fujioka92,KremerESOP05,Bernat06}), $\spr{.}{.}$ is a pairing
constructor, and $\enc{.}{.}$ is the Dolev-Yao symmetric encryption
function.  Additionally, we also assume a possibly empty equational
theory $E$, whose signature is denoted with $\Sigma_E.$ We require
that $\Sigma_C \cap \Sigma_E = \emptyset.$\footnote{This restriction
means that intruder theory such as homomorphic encryption is excluded.
Nevertheless, it still covers a wide range of intruder theories.}
Function symbols (including
constructors) are ranged over by $f$, $g$ and $h$.  The equational
theory $E$ may contain at most one associative-commutative function
symbol, which we denote with $\oplus$, obeying the standard
associative and commutative laws.  We restrict ourselves to equational
theories which can be represented by terminating and confluent rewrite
systems, modulo the associativity and commutativity of $\oplus.$ We
consider the set of messages generated by the following grammar
$$
\begin{array}{ll}
M, N := & a \mid x \mid \pub M \mid \sign M N \mid \blind M N 
\\      & \;\;\; \mid \spr M N \mid \enc M N \mid f(M_1,\ldots,M_k).
\end{array}
$$
The message $\pub M$ denotes the public key generated from a private key $M$; 
$\sign M N$ denotes a message $M$ signed with a private key $N$;
$\blind M N$ denotes a message $M$ encrypted with $N$ using a 
special blinding encryption; $\spr M N$ denotes a pair of messages;
and $\enc M N$ denotes a message $M$ encrypted with a key $N$ using
a Dolev-Yao symmetric encryption. 
The blinding encryption has a special property that it commutes with
the $\signsym$ operation, i.e., one can ``unblind'' a signed blinded
message $\sign {\blind M r} k$ using the blinding key $r$  
to obtain $\sign M k.$ This aspect of the blinding encryption is 
reflected in its elimination rules, as we shall see later.
We denote with $V(M)$ the set of variables occurring in $M$. A term $M$ is {\em ground}
if $V(M) = \emptyset.$ 
We shall be mostly concerned with ground terms, so unless
stated otherwise, we assume implicitly that terms are ground (the only exception 
is Proposition~\ref{prop:var-abs} and Proposition~\ref{prop:var-abs2}).

We shall use several notions of equality so we distinguish them using the following
notation: we shall write $M = N$ to denote syntactic equality, $M
\equiv N$ to denote equality modulo associativity and commutativity
(AC) of $\oplus$, and $M \approx_T N$ to denote equality modulo a
given equational theory $T$. We shall sometimes omit the subscript in
$\approx_T$ if it can be inferred from context.

Given an equational theory $E$, we denote with $R_E$ the set of
rewrite rules for $E$ (modulo AC). We write $M \to_{R_E} N$ when 
$M$ rewrites to $N$ using one application of a rewrite rule in $R_E$. The definition
of rewriting modulo AC is standard and is omitted here. 
The reflexive-transitive closure of $\to_{R_E}$  is denoted with $\rightarrow_{R_E}^*.$ 
We shall often remove the subscript $R_E$ when no confusion arises.
A term $M$ is in {\em $E$-normal form} if $M \not \to_{R_E} N$ for any $N.$
We write $\normE M E$ to denote the normal form of $M$ with respect to 
the rewrite system $R_E$, modulo commutativity and associativity of $\oplus$.
Again, the index $E$ is often omitted when it is clear which equational theory we
refer to.
This notation extends straightforwardly to sets, e.g., $\norm \Gamma$ denotes the set obtained
by normalising all the elements of $\Gamma.$
A term $M$ is said to be {\em headed by } a symbol $f$ if $M = f(M_1,\ldots,M_k)$.
$M$ is {\em guarded}
if it is either a name, a variable, or a term headed by a constructor. 
A term $M$ is an {\em $E$-alien term} if $M$ is headed by a symbol $f \not \in \Sigma_E.$
It is a {\em pure $E$-term} if it contains only symbols from $\Sigma_E$, names and variables.

A {\em context} is a term with holes. We denote with $C^k[]$ a context with $k$-hole(s).
When the number $k$ is not important or can be inferred from context, we shall write
$C[\ldots]$ instead. Viewing a context $C^k[]$ as a tree, each hole in the context
occupies a unique position among the leaves of the tree. We say that a hole occurrence
is the $i$-th hole of the context $C^k[]$ if it is the $i$-th hole encountered in
an inorder traversal of the tree representing $C^k[].$
An $E$-context is a context formed using only the function symbols in $\Sigma_E.$ 
We write $C[M_1,\ldots,M_k]$ to denote the term resulting 
from replacing the holes in the $k$-hole context $C^k[]$ with $M_1, \ldots, M_k,$
with $M_i$ occuping the $i$-th hole in $C^k[].$

\paragraph{Natural deduction and sequent systems}

The standard formulation of the judgment $\Gamma \vdash M$ is usually given in
terms of a natural-deduction style inference system, as shown in
Figure~\ref{fig:nat}. We shall refer to this proof system as $\Nscr$
and write $\Gamma \Vdash_\Nscr M$ if $\Gamma \vdash M$ is derivable in $\Nscr.$
The deduction rules for Dolev-Yao encryption is standard and can be found
in the literature, e.g., \cite{Boreale01ICALP,Comon-Lundh03LICS}. 
The blind signature rules are taken from the formulation given by Bernat and 
Comon-Lundh~\cite{Bernat06}. Note that the rule $\signsym_E$ assumes
implicitly that signing a message hides its contents. An alternative rule
without this assumption would be
$$
\infer[]
{\Gamma \vdash M}
{\Gamma \vdash \sign M K}
$$
The results of the paper also hold, with minor modifications, if we adopt this rule. 

\begin{figure}[t]
  \begin{tabular}[c]{l@{\extracolsep{0.8cm}}ll}
    $\infer[id]{\Gamma \vdash M}
               {M \in \Gamma}$
  & $\infer[e_E]{\Gamma \vdash M}
                {\Gamma \vdash \enc M K & \Gamma \vdash K}$
  & $\infer[e_I]{\Gamma \vdash \enc M K}
                {\Gamma \vdash M & \Gamma \vdash K}$
\\[1em]
    $\infer[p_E]{\Gamma \vdash M}
                {\Gamma \vdash \spr M N}$
  & $\infer[p_E]{\Gamma \vdash N}
                {\Gamma \vdash \spr M N}$
  & $\infer[p_I]{\Gamma \vdash \spr M N}
                {\Gamma \vdash M & \Gamma \vdash N}$
\\[1em]
  \multicolumn{2}{c}{
      $\infer[\signsym_E]{\Gamma \vdash M}
                         {\Gamma \vdash \sign M K & \Gamma \vdash \pub K}$
    }
  &  $\infer[\signsym_I]{\Gamma \vdash \sign M K}
                       {\Gamma \vdash M & \Gamma \vdash K}$
\\[1em]
    \multicolumn{2}{c}{
    $\infer[\blindsym_{E1}]{\Gamma \vdash M}
                          {\Gamma \vdash \blind M K & \Gamma \vdash K}$
    }
  &
    $\infer[\blindsym_I]{\Gamma \vdash \blind M K}
                        {\Gamma \vdash M & \Gamma \vdash K}$
\\[1em]
    \multicolumn{3}{c}{
     $\infer[\blindsym_{E2}]{\Gamma \vdash \sign M K}
                       {\Gamma \vdash \sign {\blind M R} K & \Gamma \vdash R}$
    }
\\[1em]
    \multicolumn{2}{c}{
    $\infer[f_I, \hbox{ where } f \in \Sigma_E]
           {\Gamma \vdash f(M_1,\ldots,M_n)}
           {\Gamma \vdash M_1 & \cdots & \Gamma \vdash M_n}$
    }
  &
    \multicolumn{1}{c}{
    $\infer[\approx, \hbox{ where $M \approx_{E} N$}]
                          {\Gamma \vdash M}
                          {\Gamma \vdash N}$
    }
  \end{tabular}
\caption{System $\Nscr$: a natural deduction system for intruder deduction}
\label{fig:nat}
\end{figure}

A sequent $\Gamma \vdash M$ is in {\em normal form} if $M$ and all the terms in $\Gamma$ are
in normal form. Unless stated otherwise, in the following we assume that sequents
are in normal form.
The sequent system for intruder deduction, under the equational theory $E$, 
is given in Figure~\ref{fig:msg}. We refer to this sequent system as $\Sscr$
and write $\Gamma \Vdash_\Sscr M$ to denote the fact that the sequent
$\Gamma \vdash M$ is derivable in $\Sscr.$

Unlike natural deduction rules, sequent rules also allow introduction of terms on
the left hand side of the sequent. 
The rules $p_L,$ $e_L,$ $\signsym_L,$ $\blindsym_{L1},$ $\blindsym_{L2},$ and  $gs$ are called
{\em left introduction rules} (or simply {\em left rules}),
and the rules $p_R, e_R, \signsym_R, \blindsym_R$ are called {\em right introduction
rules} (or simply, {\em right rules}).
Notice that the rule $gs$ is very similar to $cut$, except that we have the
proviso that $A$ is a subterm of a term in the lower sequent. This is sometimes
called {\em analytic cut} in the proof theory literature. Analytic cuts are 
not problematic as far as proof search is concerned, since it still obeys the
sub-formula property. 

We need the rule $gs$ because we do not have introduction rules for
function symbols in $\Sigma_E$, in contrast to natural deduction. This
rule is needed to ``abstract'' $E$-alien subterms in a sequent (in the
sense of the variable abstraction technique common in unification theory, see
e.g., \cite{Schmidt-Schauss89,Baader96JSC}), which is needed to prove
that the cut rule is redundant.  For example, let $E$ be a theory containing 
only the associativity and the commutativity axioms for $\oplus$. 
Then the sequent ~ $ a, b \vdash
\spr a b \oplus a $ should be provable without cut. Apart from the
$gs$ rule, the only other way to prove this is by using the $id$
rule. However, $id$ is not applicable, since no $E$-context
$C[...]$ can obey $C[a,b] \approx
\spr a b \oplus a$ 
because
$E$-contexts can contain only symbols from
$\Sigma_E$ and thus cannot contain $\spr . .$. 
Therefore we need to ``abstract'' the term $\spr a
b$ in the right hand side, via the $gs$ rule:
$$
\infer[gs]
{a, b \vdash \spr a b \oplus a}
{
  \infer[p_R]
  {a, b \vdash \spr a b}
  {\infer[id]{a,b\vdash a}{} & \infer[id]{a,b \vdash b}{}}
 & 
  \infer[id]
  {a, b, \spr a b \vdash \spr a b \oplus a}
  {}
}
$$
The third $id$ rule instance (from the left) is valid because we have
$C[\spr a b, a] \equiv \spr a b \oplus a$, where
$C[.,.] = [.] \oplus [.].$ 

\begin{figure}[t]
  \begin{tabular}[c]{ll}
   $\infer[id]{\Gamma \vdash M}
              {
              \begin{array}{c}
               M \approx_{E} C[M_1,\ldots,M_k] \\
                \hbox{$C[\ ]$ an $E$-context, and $M_1,\ldots,M_k \in \Gamma$}
              \end{array}
              }$
  & $\infer[cut]{\Gamma \vdash T}
                {\Gamma \vdash M
                 & \Gamma, M \vdash T}$
\\[1em]
  $\infer[p_L]{\Gamma, \spr M N \vdash T}
              {\Gamma, \spr M N, M, N \vdash T}$
  & $\infer[p_R]{\Gamma \vdash \spr M N}
                {\Gamma \vdash M & \Gamma \vdash N}$
\\[1em]
  $\infer[e_L]{\Gamma, \enc M K \vdash N}
              {\Gamma, \enc M K \vdash K
                & \Gamma, \enc M K, M, K \vdash N}$
  & $\infer[e_R]{\Gamma \vdash \enc M K}
                {\Gamma \vdash M & \Gamma \vdash K}$
\\[1em]
  $\infer[\signsym_L, K \equiv L]{\Gamma, \sign M K, \pub L \vdash N}
                                 {\Gamma, \sign M K, \pub L, M \vdash N}$
  & $\infer[\signsym_R]{\Gamma \vdash \sign M K}
                       {\Gamma \vdash M & \Gamma \vdash K}$
\\[1em]
  $\infer[\blindsym_{L1}]{\Gamma, \blind M K \vdash N}
                        {\Gamma, \blind M K \vdash K 
                         & \Gamma, \blind M K, M, K \vdash N}$
  & $\infer[\blindsym_R]{\Gamma \vdash \blind M K}
                       {\Gamma \vdash M & \Gamma \vdash K}$
\\[1em]
  \multicolumn{2}{c}{
    $\infer[\blindsym_{L2}]{\Gamma, \sign {\blind M R} K \vdash N}
                          {\Gamma, \sign {\blind M R} K \vdash R
                           &
                           \Gamma, \sign {\blind M R} K, \sign M K, R \vdash N}$
  }
\\[1em]
  \multicolumn{2}{c}{
   $\infer[gs, \hbox{$A$ is a guarded subterm of $\Gamma \cup \{M\}$}]
           {\Gamma \vdash M}
           {\Gamma \vdash A & \Gamma, A \vdash M}$
  }
\end{tabular}
\caption{System $\Sscr$: a sequent system for intruder deduction. }
\label{fig:msg}
\end{figure}

Provability in the natural deduction system and in the sequent system
are related via the standard translation, i.e., right rules in sequent calculus
correspond to introduction rules in natural deduction and left rules corresponds
to elimination rules. The straightforward 
translation from natural deduction to sequent calculus uses the cut rule.

\begin{proposition}
\label{prop:S-equal-N}
The judgment $\Gamma \vdash M$ is provable in the natural deduction system $\Nscr$ 
if and only if $\norm \Gamma \vdash \norm M$ is provable in the sequent system $\Sscr$.
\end{proposition}

\section{Cut elimination for $\Sscr$}
\label{sec:cutelim}

We now show that the cut rule is redundant for $\Sscr$. 
\begin{definition}
An inference rule $R$ in a proof system ${\cal D}$ is {\em admissible for ${\cal D}$} if
for every sequent $\Gamma \vdash M$ derivable in ${\cal D}$, there is a derivation
of the same sequent in ${\cal D}$ without instances of $R$.
\end{definition}
The {\em cut-elimination} theorem for $\Sscr$ 
states that the cut rule is admissible for $\Sscr$. Before we proceed
with the main cut elimination proof, we first prove a basic property
of equational theories and rewrite systems, which is concerned with a
technique called {\em variable abstraction} \cite{Schmidt-Schauss89,Baader96JSC}.

Given derivation $\Pi$, the {\em height} of the derivation, denoted by $|\Pi|$, 
is the length of a longest branch in $\Pi.$ Given a normal term $M$, 
the {\em size} $|M|$ of $M$ is the number of function 
symbols, names and variables appearing in $M.$

In the following, we consider slightly more general equational
theories than in the previous section:  
each $AC$ theory $E$ can be a theory obtained from a disjoint combination of 
$AC$ theories $E_1, \ldots, E_k$, where each $E_i$ has at most
one AC operator $\oplus_i.$ This allows us to reuse the results for a 
more general case later.

\begin{definition}
  Let $E$ be a disjoint combination of AC convergent theories
  $E_1,$ $\dots,$ $E_n$.  A term $M$ is a {\em quasi-$E_i$ term} if every
  $E_i$-alien subterm of $M$ is in $E$-normal form.
\end{definition}

For example, let $E = \{h(x,x) \approx x \}$.  
Then $h(\spr a b, c)$ is a quasi $E$-term, whereas $h(\spr a b, \spr
{h(a,a)} b)$ is not, since its $E$-alien subterm $\spr {h(a,a)} b$ is not
in its $E$-normal form $\spr{a}{b}$. Obviously, any $E$ normal term is a quasi $E_i$ term.

In the following, given an equational theory $E$, we assume the
existence of a function $v_E$, which assigns a variable from $\varset$
to each ground term such that $v_E(M) = v_E(N)$ if and only if $M
\approx_E N.$ In other words, $v_E$ assigns a unique variable to each
equivalence class of ground terms induced by $\approx_E.$

\begin{definition}
Let $E$ be an equational theory obtained by disjoint combination of 
AC theories $E_1, \ldots, E_n$. 
The {\em $E_i$ abstraction function} $F_{E_i}$ is a function 
mapping ground terms to pure $E_i$ terms, defined recursively as follows:
$$
F_{E_i}(u) = 
\left\{
\begin{array}{ll}
u, & \hbox{ if $u$ is a name, }\\
f(F_{E_i}(u_1), \ldots, F_{E_i}(u_k)), & \hbox{if $u=f(u_1,\ldots,u_k)$ and $f \in \Sigma_{E_i}$,} \\
v_E(u), & \hbox{ otherwise.}
\end{array}
\right.
$$
\end{definition}

It can be easily shown that the function $F_{E_i}$ preserves the equivalence relation
$\equiv$. That is, if $M \equiv  N$ then $F_{E_i}(M) \equiv F_{E_i}(N)$.

\begin{proposition}
\label{prop:var-abs}
Let $E$ be a disjoint combination of $E_1,\ldots,E_n$. 
If $M$ is a quasi $E_i$ term and $M \to^*_{R_E} N$, then
$N$ is a quasi $E_i$ term and $F_{E_i}(M) \to^*_{R_E} F_{E_i}(N).$
\end{proposition}

\begin{proposition}
\label{prop:var-abs2}
Let $E$ be a disjoint combination of $E_1,\ldots,E_n$. 
If $M$ and $N$ are quasi $E_i$ terms and
$F_{E_i}(M) \to^*_{R_E} F_{E_i}(N)$, then 
$M \to^*_{R_E} N.$
\end{proposition}

We now show some important proof transformations needed to prove cut elimination,
i.e., in an inductive argument to reduce the size of cut terms. 
In the following, when we write that a sequent $\Gamma \vdash M$ is derivable, 
we mean that it is derivable in
the proof system $\Sscr$, with a fixed AC theory $E$.

\begin{lemma}
\label{lm:equiv}
Let $\Pi$ be a derivation of $M_1, \ldots, M_k \vdash N.$ Then for any
$M_1',$ $\ldots,$ $M_k'$ and $N'$
such that $M_i \equiv M_i'$ and $N \equiv N'$, there is a derivation $\Pi'$
of $M_1', \ldots, M_k' \vdash N'$ such that $|\Pi| = |\Pi'|.$
\end{lemma}

\begin{lemma}
\label{lm:decomp1}
Let $X$ and $Y$ be terms in normal form and let $f$ be a binary
constructor.  If $\Gamma, f(X,Y) \vdash M$ is cut-free derivable, then
so is $\Gamma, X, Y \vdash M$.
\end{lemma}

The more interesting case in the proof of Lemma~\ref{lm:decomp1}
is when  $\Gamma, f(X,Y) \vdash M$ is proved by an application of the $id$
rule where $f(X,Y)$ is active. That is, we have
$
C[f(X,Y),M_1, \ldots, M_k] \approx_E M, 
$
where $M_1,\dots,M_k \in \Gamma$, 
for some $E$-context $C[..].$ Since $M$ is in normal form, we have
\begin{equation}
\label{eq1}
C[f(X,Y), M_1, \ldots, M_k] \rightarrow^* M.
\end{equation}
There are two cases to consider in the construction of a proof for
$\Gamma, X, Y \vdash M.$ If $f(X,Y)$ 
occurs as a subterm of $M$ or $\Gamma$, then we simply apply the $gs$
rule (bottom up) to abstract the term $f(X,Y)$ and then apply the
$id$ rule. 
Otherwise,  we use the variable abstraction techniques 
(Proposition~\ref{prop:var-abs} and  Proposition~\ref{prop:var-abs2}) 
to abstract $f(X,Y)$ from the rewrite steps (\ref{eq1}) above,
and then replace its abstraction with $X$ to obtain: 
$
C[X,M_1, \ldots, M_k] \rightarrow^* M.
$
That is, the $id$ rule is applicable to the sequent $\Gamma, X, Y \vdash M$,
with $X$ taking the role of $f(X,Y)$. 

\begin{lemma}
\label{lm:decomp2}
Let $X_1,\ldots, X_k$ be normal terms and let 
$\Pi$ be a cut-free derivation of $\Gamma, \norm{f(X_1,\ldots,X_k)} \vdash M$, where $f \in \Sigma_E.$
Then there exists a cut-free derivation $\Pi'$ of $\Gamma, X_1,\ldots,X_k \vdash M.$ 
\end{lemma}

\begin{lemma}
\label{lm:decomp3}
Let $M_1, \ldots, M_k$ be terms in normal form and let $C[\ldots]$ be
a $k$-hole $E$-context. If $\Gamma, \norm{C[M_1,\ldots,M_k]} \vdash M$
is cut-free derivable, then so is $\Gamma, M_1, \ldots, M_k \vdash M$.
\end{lemma}

One peculiar aspect of the sequent system $\Sscr$ is that in the
introduction rules for encryption functions (including blind
signatures), there is no switch of polarities for the encryption
key. For example, in the introduction rule for $\enc M K$, both on the
left and on the right, the key $K$ appears on the right hand
side of a premise of the rule. This means that there is no exchange
of information between the left and the right hand side of sequents,
unlike, say, typical implication rules in logic. This gives rise to an
easy cut elimination proof, where we need only to measure the
complexity of the left premise of a cut in determining the cut
rank. 

\begin{theorem}
\label{thm cut elim}
The cut rule is admissible for $\Sscr$.
\end{theorem}

\section{Normal derivations and decidability}
\label{sec:normal}

We now turn to the question of the decidability of the deduction problem
$\Gamma \vdash M.$ This problem is known already for several AC theories, 
e.g., exclusive-or, abelian groups and their extensions with a homomorphism
axiom~\cite{Comon-Lundh03LICS,Chevalier03LICS,Delaune06ICALP,Delaune06IPL,Abadi06TCS}. 
What we would like to show here is how the decidability
result can be reduced to a more elementary decision 
problem, defined as follows.

\begin{definition}
Given an equational theory $E$, the {\em elementary deduction problem} for $E$, written $\Gamma \Vdash_E M$, 
is the problem of deciding whether the $id$ rule is applicable
to the sequent $\Gamma \vdash M$ (by checking 
whether there exists an $E$-context
$C[\ldots]$ and terms $M_1,\ldots, M_k \in \Gamma$ such that
$C[M_1,\ldots,M_k] \approx_{E} M$).
\end{definition}

Note that as a consequence of Proposition~\ref{prop:var-abs} and Proposition~\ref{prop:var-abs2},
in checking elementary deducibility, it is enough to consider the pure $E$ equational
problem where all $E$-alien subterms are abstracted, i.e., we have
$$
C[M_1, \ldots, M_k] \approx_{E} M
\qquad
\mbox{iff}
\qquad
C[F_E(M_1), \ldots, F_E(M_k)] \approx_E F_E(M).
$$
Our notion of elementary deduction corresponds roughly to the notion of ``recipe''
in \cite{Abadi06TCS}, but we note that the notion of a recipe is a stronger one, since it
bounds the size of the equational context.

The cut free sequent system does not strictly speaking enjoy the ``sub-formula'' property,
i.e., in $\blindsym_{L2}$, the premise sequent has a term which is not a subterm of
any term in the lower sequent. However, it is easy to see that, reading the rules bottom up,
we only ever introduce terms which are smaller than the terms in the lower sequent.
Thus a naive proof search strategy which non-deterministically tries all applicable rules
and avoids repeated sequents will eventually terminate. 
This procedure is of course rather expensive. 
We show that we can obtain a better complexity result by analysing the structures of cut-free derivations. 
Recall that the rules $p_L, e_L, \signsym_L, \blindsym_{L1}, \blindsym_{L2}$ and $gs$
are called left rules (the other rules are right rules). 

\begin{definition}
A cut-free derivation $\Pi$ is said to be a {\em normal derivation} if 
it satisfies the following conditions: 
\begin{enumerate}
\item no left rule appears above a right rule;
\item no left rule appears immediately above the left-premise of a
  branching left rule (i.e., all left rules except $p_L$ and
  $\signsym_L$).
\end{enumerate}
\end{definition}

\begin{proposition}
\label{prop:norm}
If $\Gamma \vdash M$ is derivable then it has a normal derivation. 
\end{proposition}

In a normal derivation, the left branch of a branching left rule is
provable using only right rules and $id$. This means that we can
represent a normal derivation as a sequence (reading the proof bottom-up) 
of sequents, each of which is obtained from the previous
one by adding terms composed of subterms of the previous sequent, with
the proviso that certain subterms can be constructed using
right-rules.  Let us denote with $\Gamma \Vdash_\Rscr M$ the fact that
the sequent $\Gamma \vdash M$ is provable using only the right rules
and $id$.  This suggests a more compact deduction system for intruder
deduction, called system $\Lscr$, given in Figure~\ref{fig:linear}.

\begin{figure}[t]
  \begin{tabular}[c]{ll}
   $\infer[r]{\Gamma \vdash M}
             {\Gamma \Vdash_\Rscr M}$
  &
   \qquad $\infer[le, \hbox{ where $\Gamma, \enc M K \Vdash_\Rscr K$}]
              {\Gamma, \enc M K \vdash N}
              {\Gamma, \enc M K, M, K \vdash N}$
\\[1em]
   $\infer[lp]{\Gamma, \spr M N \vdash T}
              {\Gamma, \spr M N, M, N \vdash T}$
   &
   \qquad $\infer[\signsym, K \equiv L]
              {\Gamma, \sign M K, \pub L \vdash N}
              {\Gamma, \sign M K, \pub L, M \vdash N}$
\\[1em]
  \multicolumn{2}{l}{
   $\infer[\blindsym_1, \hbox{ where $\Gamma, \blind M K \Vdash_\Rscr K$ }]
             {\Gamma, \blind M K \vdash N}
             {\Gamma, \blind M K, M, K \vdash N}$
  }
\\[1em]
  \multicolumn{2}{l}{
   $\infer[\blindsym_2,]
             {\Gamma, \sign {\blind M R} K \vdash N}
             {\Gamma, \sign {\blind M R} K, \sign M K, R \vdash N}$
  }
\\[1em] 
  \multicolumn{2}{l}{\qquad where $\Gamma, \sign {\blind M R} K \Vdash_\Rscr R.$} 
\\[1em]
  \multicolumn{2}{l}{
  $\infer[ls, 
           \hbox{where $A$ is a guarded subterm of $\Gamma \cup \{M\}$ 
                 and $\Gamma \Vdash_\Rscr A.$}]
            {\Gamma \vdash M}
            {\Gamma, A \vdash M} $
  }
  \end{tabular}
\caption{System $\Lscr$: a linear proof system for intruder deduction.}
\label{fig:linear}
\end{figure}

\begin{proposition}
\label{prop:normal}
Every sequent $\Gamma \vdash M$ is provable in $\Sscr$ if and only if it
is provable in $\Lscr.$
\end{proposition}

We now show that the decidability of the deduction problem $\Gamma \Vdash_\Sscr M$
can be reduced to decidability of elementary deduction problems. 
We consider a representation of terms as directed acyclic graphs (DAG), 
with maximum sharing of subterms. 
Such a representation is quite standard and can be found in, e.g., \cite{Abadi06TCS}, so
we will not go into the details here. 

In the following, we denote with $st(\Gamma)$ the set of subterms of the terms in $\Gamma.$
In the DAG representation of $\Gamma$, the number of distinct nodes in the DAG representing 
distinct subterms of $\Gamma$ co-incides with the cardinality of $st(\Gamma).$ 
A term $M$ is a {\em proper subterm} of $N$ if $M$ is a subterm of $N$ and $M \not = N.$
We denote with $pst(\Gamma)$ the set of proper subterms of $\Gamma$,
and we define 
$$
sst(\Gamma) = \{\sign M N \mid M, N \in pst(\Gamma) \}.
$$
The {\em saturated set} of $\Gamma$, written $St(\Gamma)$, is the set
$$
St(\Gamma) = \Gamma \cup pst(\Gamma) \cup sst(\Gamma).
$$
The cardinality of $St(\Gamma)$ is at most quadratic in the size of
$st(\Gamma)$. If $\Gamma$ is represented as a DAG, one can 
compute the DAG representation of $St(\Gamma)$ in polynomial time, with only
a quadratic increase of the size of the graph. 
Given a DAG representation of $St(\Gamma \cup \{M\})$, 
we can represent a sequent $\Gamma \vdash M$ by associating 
each node in the DAG with a tag which indicates 
whether or not the term represented by the subgraph rooted at that node 
appears in $\Gamma$ or $M$. 
Therefore, in the following complexity results for deducibility problem 
$\Gamma \Vdash_S M$ (for some proof system $S$), we assume that 
the input consists of the DAG representation of the saturated set 
$St(\Gamma \cup \{M\})$, together with approriate tags in the nodes.
Since each tag takes only a fixed amount of space (e.g., a two-bit data
structure should suffice), we shall state the complexity result
w.r.t. the size of $St(\Gamma \cup \{M\}).$

\begin{definition}
Let $\Gamma \Vdash_{\Dscr} M$ be a deduction
problem, where $\Dscr$ is some proof system, and let $n$ be the size of $St(\Gamma \cup \{M\}).$
Let $E$ be the equational theory associated with $\Dscr$. Suppose that
the elementary deduction problem in $E$ has complexity $O(f(m)),$ where $m$ is
the size of the input. 
Then the problem $\Gamma \Vdash_{\Dscr} M$ is said to be {\em polynomially reducible} 
to the elementary deduction problem $\Vdash_E$
if it has complexity $O(n^k \times f(n))$ for some constant $k.$
\end{definition}

A key lemma in proving the decidability result is 
the following invariant property of linear proofs.
\begin{lemma}
\label{lm:St-inv}
Let $\Pi$ be an $\Lscr$-derivation of $\Gamma \vdash M.$
Then for every sequent $\Gamma' \vdash M'$ occurring in $\Pi$,
we have
$\Gamma' \cup \{M'\} \subseteq St(\Gamma \cup \{M\}).$
\end{lemma}
The existence of linear size proofs then follows from the above
lemma.
\begin{lemma}
\label{lm:linear}
If there is an $\Lscr$-derivation of $\Gamma \vdash M$ then there is an $\Lscr$-derivation
of the same sequent whose length is at most $|St(\Gamma \cup \{M\})|.$
\end{lemma}

Another useful observation is that the left-rules of $\Lscr$ are {\em invertible};
at any point in proof search, we do not lose provability by applying any 
left rule. Polynomial reducibility of $\Vdash_\Lscr$ to $\Vdash_E$ can
then be proved by a deterministic proof search strategy which systematically tries
all applicable rules. 

\begin{theorem}
\label{thm:deducibility}
The decidability of the relation $\Vdash_\Lscr$ is polynomially reducible
to the decidability of elementary deduction $\Vdash_E.$
\end{theorem}

Note that in the case where the theory $E$ is empty, we obtain a
\textsc{ptime} decision procedure for intruder deduction with blind
signatures.

\section{Some example theories}
\label{sec:examples}

We now consider several concrete AC convergent theories that are often used
in reasoning about security protocols. Decidability of intruder deduction
under these theories has been extensively studied 
\cite{Comon-Lundh03LICS,Chevalier03LICS,Abadi06TCS,Delaune06ICALP,Lafourcade07IC,Cortier07LPAR}. 
These results can be broadly categorized into those with explicit
pairing and encryption constructors, e.g., \cite{Comon-Lundh03LICS,Lafourcade07IC}, and those where the constructors
are part of the equational theories, e.g., \cite{Abadi06TCS,Cortier07LPAR}. 
For the latter, one needs explicit decryption operators with, e.g., an equation like
$
dec(\enc M N, N) \approx M.
$
Decidability results for these deduction problems 
are often obtained by separating elementary deducibility from
the general deduction problem. This is obtained by studying some form of normal
derivations in a natural deduction setting. Such a reduction, as has been shown
in the previous section, applies to our calculus in a more systematic fashion. 

\paragraph{Exclusive-or.}
The signature of this theory consists
of a binary operator $\oplus$ and a constant $0.$ The theory is given by the axioms
of associativity and commutativity of $\oplus$ together with the axiom $x \oplus x \approx 0$
and $x \oplus 0 \approx x.$ This theory can be turned into an AC convergent rewrite system
with the following rewrite rules:
$$
x \oplus x \to 0 \qquad \hbox{ and } \qquad x \oplus 0 \to x.
$$
Checking $\Gamma \Vdash_E M$ can be done in PTIME, as shown in, e.g., \cite{Chevalier03LICS}.

\paragraph{Abelian groups.}

The exclusive-or theory is an instance of Abelian groups, where the inverse of
an element is the element itself. The more general case of Abelian groups includes
an inverse operator, denoted with $I$ here. The equality theory for Abelian groups is given by the axioms
of associativity and commutativity, plus the theory 
$
\{ x \oplus 0 \approx 0, x \oplus I(x) \approx 0 \}.
$
The equality theory of Abelian groups can be turned into a rewrite system modulo AC by orienting
the above equalities from left to right, in addition to the following rewrite rules:
$$
I(x\oplus y) \to I(x) \oplus I(y) \qquad I(I(x)) \to x \qquad I(0) \to 0.
$$
One can also obtain an AC convergent rewrite system for an extension of Abelian groups
with a homomorphism axiom involving a unary operator $h$: $h(x \oplus y) = h(x) \oplus h(y).$
In this case, the rewrite rules above need to be extended with 
$$
h(x \oplus y) \to h(x) \oplus h(y) \qquad h(0) \to 0 \qquad h(I(x)) \to I(h(x)). 
$$
Decidability of elementary deduction under Abelian groups (with homomorphism)
can be reduced to solving a system of linear equations over some semirings 
(see \cite{Delaune06IPL} for details).

\section{Combining disjoint convergent theories}
\label{sec:comb}

We now consider the intruder deduction problem under a convergent AC theory
$E$, which is obtained from the union of pairwise disjoint convergent 
AC theories $E_1, \ldots, E_n.$ Each theory $E_i$ may contain an associative-commutative
binary operator, which we denote with $\oplus_i.$
We show that the intruder deduction problem
under $E$ can be reduced to the elementary deduction problem of each $E_i.$

Given a term $M = f(M_1,\ldots, M_k)$, where $f$ is a function symbol (i.e., a constructor,
an equational symbol or $\oplus$), the terms $M_1, \ldots, M_k$ are called the
{\em immediate subterms} of $M.$ 
Given a term $M$ and a subterm occurrence $N$ in $M$, we say that $N$ is a 
{\em cross-theory subterm} of $M$ if $N$ is 
headed with a symbol $f \in \Sigma_{E_i}$ and it is an immediate subterm
of a subterm in $M$ which is headed by a symbol $g \in \Sigma_{E_j}$, where $i \not = j.$
We shall also refer to $N$ as an {\em $E_{ij}$-subterm} of $M$ when we need to be explicit
about the equational theories involved.

Throughout this section, we consider a sequent system $\Dscr$, whose rules
are those of $\Sscr$, but with $id$ replaced by 
the rule below left and
with the addition of the rule below right
where $N$ is a cross-theory subterm of some term in $\Gamma \cup \{M\}$:
$$
\infer[id_{E_i}]
{\Gamma \vdash M}
{
\begin{array}{c}
M \approx_{E} C[M_1,\ldots,M_k] \\
\hbox{$C[\ ]$ an $E_i$-context, and $M_1,\ldots,M_k \in \Gamma$}
\end{array}
}
\qquad
\infer[cs]
{\Gamma \vdash M}
{\Gamma \vdash N & \Gamma, N \vdash M}
$$

The analog of Proposition~\ref{prop:S-equal-N} holds for $\Dscr$. Its proof is a straightforward
adaptation of the proof of Proposition~\ref{prop:S-equal-N}.

\begin{proposition}
\label{prop:D-equal-N}
The judgment $\Gamma \vdash M$ is provable in the natural deduction system $\Nscr$, under theory $E$, 
if and only if $\norm \Gamma \vdash \norm M$ is provable in the sequent system $\Dscr$.
\end{proposition}

Cut elimination also holds for $\Dscr$. Its proof is basically the same as the proof
for $\Sscr$, since the ``logical structures'' (i.e., those concerning constructors)
are the same. The only difference is in the treatment of abstracted terms (the rules
$gs$ and $cs$). In $\Dscr$ we allow abstraction of arbitrary cross-theory subterms,
in addition to guarded subterm abstraction.
The crucial part of the proof in this case relies on the variable abstraction
technique (Proposition~\ref{prop:var-abs} and Proposition~\ref{prop:var-abs2}),
which applies to both guarded subterm and cross-theory subterm
abstraction.

\begin{theorem}
\label{thm cut elim for D}
The cut rule is admissible for $\Dscr$. 
\end{theorem}

The decidability result for $\Sscr$ also holds for $\Dscr.$ 
This can be proved with straightforward modifications of the similar proof for $\Sscr$,
since the extra rule $cs$ has the same structure as $gs$ in $\Sscr.$
It is easy to see that the same normal forms for $\Sscr$ also holds for $\Dscr,$ with $cs$
considered as a left-rule.  
It then remains to design a linear proof system for $\Dscr.$ 
We first define the notion of right-deducibility: The relation $\Gamma \Vdash_{\Rscr\Dscr} M$
holds if and only if the sequent $\Gamma \vdash M$ is derivable in $\Dscr$ using only the right rules.
We next define a linear system for $\Dscr$, called $\Lscr\Dscr$, which consists of 
the rules of $\Lscr$ defined in the previous section, but with the proviso $\Gamma \Vdash_\Rscr M$
changed to $\Gamma \Vdash_{\Rscr\Dscr} M$, and with the additional rule:
$$
\infer[lcs]
{\Gamma \vdash M}
{\Gamma, R \vdash M} 
$$
where $R$ is a cross-theory subterm of some term in $\Gamma \cup \{M\}$ and $\Gamma \Vdash_{\Rscr\Dscr} R.$

\begin{proposition}
A sequent $\Gamma \vdash M$ is provable in $\Dscr$ if and only if it
is provable in $\Lscr\Dscr.$
\end{proposition}

The notion of polynomial reducibility is slightly changed. 
Suppose each elementary deduction problem in $E_i$ is bounded by $O(f(m)).$
Let $m$ be the size of $St(\Gamma \cup \{M\}).$ 
Then the deduction problem $\Gamma \Vdash_{\Dscr} M$ is polynomially reducible
to $\Vdash_{E_1}, \ldots, \Vdash_{E_n}$ if it has complexity
$O(m^k f(m))$, for some constant $k$.

\begin{theorem}
\label{thm:D-deducibility}
The decidability of the relation $\Vdash_{\Lscr\Dscr}$ is polynomially reducible
to the decidability of elementary deductions $\Vdash_{E_1},$ $\ldots,$ $\Vdash_{E_n}$.
\end{theorem}

\section{Conclusion and related work}
\label{sec:rel}

We have shown that decidability of the intruder deduction problem,
under a range of equational theories, can be reduced to the simpler
problem of elementary deduction, which amounts to solving equations in
the underlying equational theories. This reduction is obtained in a
purely proof theoretical way, using standard techniques such as cut
elimination and permutation of inference rules.

There are several existing works in the literature that deal with
intruder deduction.  Our work is more closely related to, e.g.,
\cite{Comon-Lundh03LICS,Delaune06IPL,Lafourcade07IC}, in that we do
not have explicit destructors (projection, decryption, unblinding),
than, say, \cite{Abadi06TCS,Cortier07LPAR}. In the latter work, these
destructors are considered part of the equational theory, so in this
sense our work slightly extends theirs to allow combinations of
explicit and implicit destructors.  A drawback for the approach with
explicit destructors is that one needs to consider these destructors
together with other algebraic properties in proving decidability,
although recent work in combining decidable theories
\cite{Arnaud07frocos} allows one to deal with them modularly.
Combination of intruder theories has been considered in
\cite{Chevalier05ICALP,Arnaud07frocos,DelauneIC08}, as part of their
solution to a more difficult problem of deducibility constraints
which assumes active intruders.  In particular, Delaune, et. al.,
\cite{DelauneIC08} obtain results similar to what we have here concerning
combination of AC theories.  One difference between these works and
ours is in how this combination is derived.  Their approach is more
algorithmic whereas our result is obtained through analysis of proof
systems.

It remains to be seen whether sequent calculus, and its associated
proof techniques, can prove useful for richer theories. For certain
deduction problems, i.e., those in which the constructors interact
with the equational theory, there does not seem to be general results
like the ones we obtain for theories with no interaction with the
constructors.  One natural problem where this interaction occurs is
the theory with homomorphic encryption, e.g., like the one considered
in \cite{Lafourcade07IC}. Another interesting challenge is to see how
sequent calculus can be used to study the more difficult problem of
solving intruder deduction constraints, e.g., like those studied in
\cite{Comon-Lundh03LICS,Chevalier03LICS,Delaune06ICALP}.

\paragraph{Acknowledgement}
We thank the anonymous referees of earlier drafts of this paper for
their careful reading and helpful comments. This work has been
supported by the Australian Research Council (ARC) Discovery Project
DP0880549.

\newpage
\appendix

\section{Proofs}

\subsection{Proofs for Section~\ref{sec:intruder}}

\begin{lemma}[Weakening]
\label{lm:weak}
Let $\Pi$ be a derivation of $\Gamma \vdash M$. If $\Gamma \subseteq \Gamma'$, then there exists a derivation
$\Pi'$ of $\Gamma' \vdash M$ such that $|\Pi| = |\Pi'|$.
\end{lemma}
\begin{proof}
By induction on $|\Pi|.$ \qed
\end{proof}

\begin{lemma}
\label{lm:nd-to-seq}
If the judgment $\Gamma \vdash M$ is derivable in natural deduction system $\Nscr$ then
$\norm \Gamma \vdash \norm M$ is derivable in sequent system $\Sscr$.
\end{lemma}
\begin{proof}
Let $\Pi$ be a natural deduction derivation of $\Gamma \vdash M$. We construct a sequent derivation
$\Pi'$ of $\norm \Gamma \vdash \norm M$ by induction on $|\Pi|.$ 
The $id$ rule translates to the $id$ rule in sequent calculus; the introduction rules
for constructors translate to the right-rules for the same constructors. 
If $\Pi$ ends with the $\approx$-rule, then the premise and the conclusion of the rules
translate to the same sequent, hence $\Pi'$ is constructed by induction hypothesis. 
It remains to show the translations for elimination rules and rules concerning $f \in \Sigma_E.$
\begin{itemize}
\item Suppose $\Pi$ ends with $f_I$, for some $f \in \Sigma_E$:
$$
\infer[f_I]
{\Gamma \vdash f(M_1,\ldots,M_k)}
{
 \deduce{\Gamma \vdash M_1}{\Pi_1}
 & \cdots &
 \deduce{\Gamma \vdash M_k}{\Pi_k}
}
$$
By induction hypothesis, we have sequent derivations $\Pi_i'$ of
$\norm \Gamma \vdash \norm {M_i}$, for each $i \in \{1,\ldots,k\}$. 
Lemma~\ref{lm:weak}, applied to the $\Pi_i'$, gives us another sequent
derivation $\Pi_i''$ of $\norm \Gamma, \norm {M_1}, \ldots, \norm {M_{i-1}} \vdash \norm {M_i}.$
We note that the sequent 
$$
\norm \Gamma, \norm {M_1}, \ldots, \norm {M_k} \vdash \norm{f(M_1,\ldots,M_k)}
$$
is provable in the sequent system by an application of the $id$-rule. 
The derivation $\Pi'$ is then constructed by successive applications of
the cut rule to this sequent with $\Pi_k'', \ldots, \Pi_1'',$ where the $i$-th
cut eliminates $\norm{M_i}.$

\item Suppose $\Pi$ ends with $p_E:$
$$
\infer[p_E]
{\Gamma \vdash M}{\deduce{\Gamma \vdash \spr M N}{\Pi_1}}
$$
Note that $\norm {\spr M N} \equiv \spr {\norm M} {\norm N}.$ By induction hypothesis,
we have a sequent derivation $\Pi_1'$ of $\norm \Gamma \vdash \spr {\norm M}{\norm N}$,
and since the sequent 
$$\norm \Gamma, \spr {\norm M}{\norm N} \vdash \norm M$$ 
is derivable in sequent calculus (using an $id$ rule followed by a $p_L$-rule), 
we can use the cut rule to get a sequent derivation of $\norm \Gamma \vdash \norm M.$

\item Suppose $\Pi$ ends with $e_E:$
$$
\infer[e_E]
{\Gamma \vdash M}
{
  \deduce{\Gamma \vdash \enc M N}{\Pi_1}
  &
  \deduce{\Gamma \vdash N}{\Pi_2}
}
$$
By induction hypothesis, we have a sequent derivation $\Pi_1'$
of $\norm \Gamma \vdash \enc {\norm M} {\norm N}$ 
and a sequent derivation $\Pi_2'$ of $\norm \Gamma \vdash \norm N.$
By Lemma~\ref{lm:weak}, we have a derivation $\Pi_3$ of
$\norm \Gamma, \enc {\norm M} {\norm N} \vdash \norm N.$ We construct a sequent derivation 
for the sequent
$$
\norm \Gamma, \enc {\norm M}{\norm N}, \norm N \vdash \norm M.
$$
This can be done (in a bottom-up proof construction) by 
an application of $e_L$, followed by two applications of $id.$
Then $\Pi'$ is constructed by applying the cut rule to this sequent using $\Pi_1'$ and $\Pi_3.$

\item Suppose $\Pi$ ends with $\signsym_E$:
$$
\infer[\signsym_E]
{\Gamma \vdash M}
{
 \deduce{\Gamma \vdash \sign M K}{\Pi_1}
 &
 \deduce{\Gamma \vdash \pub K}{\Pi_2}
}
$$
By induction hypothesis, we have a sequent derivation $\Pi_1'$
and a sequent derivation $\Pi_2'$ of, respectively,
$$
\norm \Gamma \vdash \sign {\norm M} {\norm K} 
\qquad \hbox{ and } 
\qquad
\norm \Gamma \vdash \pub {\norm K}.
$$
Let $\Pi_2''$ be a derivation of 
$$
\norm \Gamma, \sign {\norm M} {\norm K}  \vdash \pub {\norm K}
$$
obtained by an application of Lemma~\ref{lm:weak} to $\Pi_2'.$
Let $\Pi_3$ be the derivation
$$
\infer[\signsym_L]
{\norm \Gamma, \sign {\norm M} {\norm K}, \pub {\norm K} \vdash \norm M}
{
 \infer[id]
 {\norm \Gamma, \sign {\norm M} {\norm K}, \pub {\norm K}, \norm M \vdash \norm M}
 {}
}
$$
Then $\Pi'$ is constructed by successive applications of $\Pi_2''$ and $\Pi_1'$
to $\Pi_3.$

\item The cases where $\Pi$ ends with $\blindsym_{E1}$ is
analogous to the case with $e_E$.

\item Suppose $\Pi$ ends with $\blindsym_{E2}$:
$$
\infer[\blindsym_{E2}]
{\Gamma \vdash \sign M K}
{\deduce{\Gamma \vdash \sign {\blind M R} K}{\Pi_1}
 & 
 \deduce{\Gamma \vdash R}{\Pi_2}
}
$$
By induction hypothesis, we have a derivation $\Pi_1'$ and a derivation
$\Pi_2'$ of, respectively, 
$$
\norm \Gamma \vdash \sign {\blind {\norm M} {\norm R}} {\norm K}
\qquad
\hbox{ and }
\qquad
\norm \Gamma \vdash \norm R.
$$
Let $\Pi_3$ be the derivation
$$
\infer[\blindsym_{L2}]
{\norm \Gamma, \sign {\blind {\norm M} {\norm R}} {\norm K}  \vdash \sign {\norm M} {\norm K}}
{
\deduce{\ldots \vdash \norm R}{\Pi_2''}
&
\infer[id]
{\ldots, \sign {\norm M} {\norm K}, \norm R \vdash \sign {\norm M} {\norm K}}
{}
}
$$
where $\Pi_2''$ is obtained from $\Pi_2'$ by weakening the sequent with
$$\sign {\blind {\norm M} {\norm R}} {\norm K}.$$
Then the derivation $\Pi'$ is constructed by a cut between $\Pi_1'$ and $\Pi_3.$
\end{itemize}
\qed
\end{proof}

\begin{lemma}
\label{lm:seq-to-nd}
If $\Gamma \vdash M$ is derivable in sequent system $\Sscr$ then
$\Gamma \vdash M$ is derivable in natural deduction system $\Nscr.$
\end{lemma}
\begin{proof}
Let $\Pi$ be a sequent derivation of $\Gamma \vdash M$. We construct a natural
deduction $\Pi'$ of $\Gamma \vdash M$ by induction on $\Pi.$
\begin{itemize}
\item The right-introduction rules for $\Sscr$ maps to the same introduction rules in $\Nscr.$
$\Pi'$ in this case is constructed straightforwardly from the induction hypothesis using
the introduction rules of $\Nscr.$
\item If $\Pi$ ends with an $id$ rule, i.e., $M \approx C[M_1,\ldots,M_k]$, for some
$M_1, \ldots, M_k \in \Gamma$ and $E$-context $C[..]$, we construct a derivation $\Pi_1$ of 
$\Gamma \vdash C[M_1,\ldots,M_k]$ by induction on the context $C[\ldots]$. 
This is easily done using the $f_I$ introduction rule in $\Nscr.$
The derivation $\Pi'$ is then constructed from $\Pi_1$ by an application of the $\approx$-rule.

\item Suppose $\Gamma = \Gamma' \cup \{\spr U V\}$ and $\Pi$ ends with $p_L:$
$$
\infer[p_L]
{\Gamma', \spr U V \vdash M}
{\deduce{\Gamma', \spr U V, U, V \vdash M}{\Pi_1}}
$$
By induction hypothesis, we have an $\Nscr$-derivation $\Pi_1'$
of $\Gamma', \spr U V, U, V \vdash M$. The derivation $\Pi'$ is constructed
inductively from $\Pi_1'$ by copying the same rule applications in $\Pi_1'$, except
when $\Pi_1'$ is either 
$$
\infer[id]
{\Gamma, U, V \vdash U}
{}
~ \hbox{ or } ~
\infer[id]
{\Gamma, U, V \vdash V}
{}
$$
in which case, $\Pi'$ is 
$$
\infer[p_E]
{\Gamma \vdash U}
{
  \infer[id]
  {\Gamma \vdash \spr U V}{}
}
~ \hbox{ and } ~
\infer[p_E]
{\Gamma \vdash V}
{
  \infer[id]
  {\Gamma \vdash \spr U V}{}
}
$$
respectively.

\item Suppose $\Gamma = \Gamma' \cup \{\enc U V\}$ and $\Pi$ ends with $e_L:$
$$
\infer[e_L]
{\Gamma', \enc U V \vdash M}
{
 \deduce{\Gamma \vdash V}{\Pi_1}
 &
 \deduce{\Gamma, U, V \vdash M}{\Pi_2}
}
$$
By induction hypothesis, we have an $\Nscr$-derivation $\Pi_1'$ of 
$\Gamma \vdash V$ and an $\Nscr$-derivation $\Pi_2'$ of
$\Gamma, U, V \vdash M.$ $\Pi'$ is then constructed inductively from $\Pi_2'$ 
by applying the same rules as in $\Pi_2'$, except when 
$\Pi_2'$ is either
$$
\infer[id]
{\Gamma, U, V \vdash U}{}
~ \hbox{ or } ~
\infer[id]
{\Gamma, U, V \vdash V}{}
$$
in which case, $\Pi'$ is, respectively,
$$
\infer[e_E]
{\Gamma \vdash U}
{
  \infer[id]{\Gamma \vdash \enc U V}{}
  &
  \deduce{\Gamma \vdash V}{\Pi_1'}
}
$$
and $\Pi_1'.$

\item Suppose $\Gamma = \Gamma' \cup \{\sign N K, \pub L\}$ and 
$\Pi$ ends with $\signsym_L$:
$$
\infer[\signsym_L]
{\Gamma', \sign N K, \pub L \vdash M}
{\deduce{\Gamma', \sign N K, \pub L, N \vdash M}{\Pi_1}}
$$
where $L \equiv K$ (hence $L \approx K$).
By induction hypothesis, we have an $\Nscr$-derivation
$\Pi_1'$ of 
$$
\Gamma', \sign N K, \pub L, N \vdash M.
$$
As in the previous case, the derivation $\Pi'$ is constructed
by imitating the rules of $\Pi_1'$, except for the following
$id$ case:
$$
\infer[id]
{\Gamma', \sign N K, \pub L, N \vdash N}
{}
$$
which is replaced by 
$$
\infer[\signsym_E]
{\Gamma, \sign N K, \pub L \vdash N}
{
 \infer[id]
 {\Gamma', \sign N K, \pub L \vdash \sign N K}
 {}
 &
 \infer[\approx]
 {\Gamma', \sign N K, \pub L \vdash \pub K}
 {
  \infer[id]
  {\Gamma', \sign N K, \pub L \vdash \pub L}
  {}
 }
}
$$

\item The case where $\Pi$ ends with $\blindsym_{L1}$ is similar
to the case with $e_L.$

\item Suppose $\Gamma = \Gamma' \cup \{\sign {\blind N R} K  \}$
and  $\Pi$ ends with $\blindsym_{L2}$:
$$
\infer[\blindsym_{L2}]
{\Gamma', \sign {\blind N R} K \vdash M}
{
\deduce{ \Gamma \vdash R}{\Pi_1}
&
\deduce{ \Gamma, \sign N K, R \vdash M}{\Pi_2}
}
$$
Similarly to the previous case, we apply induction hypothesis to both
$\Pi_1$ and $\Pi_2$, obtaining $\Pi_1'$ and $\Pi_2'$. The derivation
$\Pi'$ is constructed by imitating the rules of $\Pi_2'$, but with the
following $id$ instances:
$$
\infer[id]
{\Gamma, \sign N K, R \vdash \sign N K}
{}
\qquad
\infer[id]
{\Gamma, \sign N K, R \vdash R}{}
$$
replaced by 
$$
\infer[]
{\Gamma \vdash \sign N K}
{
 \infer[id]
 {\Gamma \vdash \sign {\blind N R} K}{}
 &
 \infer[]
 {\Gamma \vdash R}{\Pi_1'}
}
\quad
\hbox{ and }
\quad
\deduce{\Gamma \vdash R}{\Pi_1'}.
$$
\item Suppose $\Pi$ ends with $gs$:
$$
\infer[gs]
{\Gamma \vdash M}
{\deduce{\Gamma \vdash A}{\Pi_1} 
 &
 \deduce{\Gamma, A \vdash M}{\Pi_2}
}
$$
By induction hypothesis, we have an $\Nscr$-derivation
$\Pi_1'$ of $\Gamma \vdash A$ and an $\Nscr$-derivation $\Pi_2'$
of $\Gamma, A \vdash M.$ Again, as in the previous cases,
we construct $\Pi'$ inductively, on the height of $\Pi_2'$, 
by imitating the rules in $\Pi_2'$, except 
when $\Pi_2'$ ends with an instance of $id$ of the form
$$
\infer[id]
{\Gamma, A \vdash A}
{}
$$
in which case, $\Pi'$ is $\Pi_1'.$
\end{itemize}
\qed
\end{proof}

\paragraph{Proposition \ref{prop:S-equal-N}}
The judgment $\Gamma \vdash M$ is provable in the natural deduction system $\Nscr$ 
if and only if $\norm \Gamma \vdash \norm M$ is provable in the sequent system $\Sscr$.
\begin{proof}
Immediate from Lemma~\ref{lm:nd-to-seq} and Lemma~\ref{lm:seq-to-nd}. \qed
\end{proof}

\subsection{Proofs for Section~\ref{sec:cutelim}}

\begin{lemma}
\label{lm:var-abs}
Let $E$ be a disjoint combination of AC theories $E_1, \ldots, E_n$.
Let $M$ be a quasi $E_i$-term. 
If $M \to_{R_E} N$ then $N$ is also a quasi $E_i$ term
and $F_{E_i}(M) \to_{R_E} F_{E_i}(N).$
\end{lemma}
\begin{proof}
By induction on the structure of $M$:
\begin{itemize}
\item If $M$ is a name then the lemma holds vacuously.
\item Suppose $M = f(u_1,\ldots,u_k).$
There are two cases to consider:
\begin{itemize}
\item The redex is in $u_j$. This case follows straightforwardly from the induction
hypothesis and the definition of $F_{E_i}$.
\item The redex is $M$. Then there must be a rewrite rule in $R_E$ of the form
$$
C[x_1,\ldots,x_n] \rightarrow C'[x_1,\ldots,x_n]
$$
where $C[..]$ and $C'[..]$ are $E_i$-context, such that
$$
M \equiv (C[x_1,\ldots,x_l])\sigma 
\qquad
\hbox{ and }
\qquad
N \equiv (C'[x_1,\ldots,x_l])\sigma
$$
for some substitution $\sigma.$
Note that since $M$ is a quasi $E_i$ term, it follows that 
each $x_i\sigma$ is also a quasi $E_i$ term. Hence $N$ must also be 
a quasi $E_i$ term.
From the definition of $F_{E_i}$, we have the following
equality (we abbreviate $F_{E_i}$ as $F$):
$$
\begin{array}{ll}
F(M) & \equiv F(C[x_1,\ldots,x_l]\sigma) \\
& = C[F(x_1\sigma), \ldots, F(x_l\sigma)] \\
& = C[x_1,\ldots,x_l] \sigma'
\end{array}
$$
where $\sigma'$ is the substitution $\{F(x_1\sigma)/ x_1,\ldots, F(x_l\sigma)/x_l \}.$
Similarly, we can show that $F(N) \equiv C'[x_1,\ldots,x_l]\sigma'.$
Therefore, we have $F(M) \rightarrow_{R_E} F(N).$
\end{itemize}

\item Suppose $M = g(u_1,\ldots,u_k)$ and $g\not \in \Sigma_{E_i}$.
Then $M$ is an $E_i$-alien subterm of $M$, and since $M$ is a quasi $E_i$ term,
$M$ must be in $E$-normal form. Therefore no reduction is possible, hence the lemma
holds vacuously.

\end{itemize}
\end{proof}

\paragraph{Proposition \ref{prop:var-abs}}
Let $E$ be a disjoint combination of $E_1,\ldots,E_n$. 
If $M$ is a quasi $E_i$ term and $M \to^*_{R_E} N$, then
$N$ is a quasi $E_i$ term and $F_{E_i}(M) \to^*_{R_E} F_{E_i}(N).$
\begin{proof} This follows directly from Lemma~\ref{lm:var-abs}. \qed
\end{proof}

\paragraph{\bf Proposition \ref{prop:var-abs2}}
Let $E$ be a disjoint combination of $E_1,\ldots,E_n$. 
If $M$ and $N$ are quasi $E_i$ terms and
$F_{E_i}(M) \to^*_{R_E} F_{E_i}(N)$, then 
$M \to^*_{R_E} N.$
\begin{proof}
It is enough to show that this holds for one-step rewrite
$F_{E_i}(M) \to_{R_E} F_{E_i}(N).$
This can be done by induction on the structure of $M$. 
In particular, we need to show that a rewrite rule that
applies to $F_{E_i}(M)$ also applies to $M$. 
Let $x_1,\ldots,x_k$ be the free variables in $F_{E_i}(M)$.
Let $M_1,\ldots,M_k$ be normal $E$-terms such that
$v_E(M_j) = x_j$ for each $j \in \{1,\ldots,k\}$, and
$$
\sigma = \{ M_1/x_1,\ldots,M_k/x_k\}.
$$
Then we can show by induction on the structure of $M$
and $N$, and using the fact that they are quasi $E_i$-terms, 
that 
$$
F_{E_i}(M)\sigma \equiv M
\hbox{ and }
F_{E_i}(N)\sigma \equiv N.
$$
Note that for any rewrite rule in a rewrite system,
by definition, we have that all the variables free in
the right-hand side of the rule are also free in the
left-hand side.
Hence, the free variables of $F_{E_i}(N)$ are 
among the free variables in $F_{E_i}(M)$ since they are
related by rewriting.

Now suppose there is a rewrite rule in $R_E$
$$
C[x_1,\ldots,x_l] \rightarrow C'[x_1,\ldots,x_l]
$$
where $C[..]$ and $C'[..]$ are $E_i$-contexts, such
that $F_{E_i}(M) \equiv C[x_1,\ldots,x_l]\theta$
and $F_{E_i}(N) \equiv C'[x_1,\ldots,x_l]\theta$, for
some substitution $\sigma.$
Then we have
$$
M \equiv F_{E_i}(M)\sigma \equiv (C[x_1,\ldots,x_n]\theta)\sigma \equiv C[x_1,\ldots,x_n](\theta \circ \sigma)
$$
and
$$
N \equiv F_{E_i}(N)\sigma \equiv (C'[x_1,\ldots,x_n]\theta)\sigma \equiv C'[x_1,\ldots,x_n](\theta \circ \sigma).
$$
Hence we also have $M \to_{R_E} N.$
\qed
\end{proof}

\paragraph{\bf Lemma \ref{lm:equiv}}
Let $\Pi$ be a derivation of $M_1,\ldots, M_k \vdash N$. Then for any
$M_1', \ldots, M_k'$ and $N'$
such that $M_i \equiv M_i'$ and $N \equiv N'$, there is a derivation $\Pi'$
of $M_1', \ldots, M_k' \vdash N'$ such that $|\Pi| = |\Pi'|.$
\begin{proof}
By induction on $|\Pi|.$ \qed
\end{proof}

\paragraph{\bf Lemma \ref{lm:decomp1}}
Let $X$ and $Y$ be terms in normal form. 
If $\Gamma, f(X,Y) \vdash M$ is cut-free derivable, where $f$ is a binary constructor, 
then $\Gamma, X, Y \vdash M$ is also cut-free derivable.
\begin{proof}
Let $\Pi$ be a cut-free derivation of $\Gamma, f(X,Y) \vdash M$. We construct a cut-free derivation
$\Pi'$ of $\Gamma, X, Y \vdash M$ by induction on $|f(X,Y)|$ with subinduction on $|\Pi|.$ 
The only non-trivial cases are when $\Pi$ ends with $\blindsym_{L2}$, acting on $f(X,Y)$,
and when $\Pi$ ends with $id$ and $f(X,Y)$ is used in the rule. We examine these cases
in more details below.

\begin{itemize}
\item Suppose $\Pi$ ends with $\blindsym_{L2}$, acting on $f(X,Y)$, i.e.,
$f = \signsym$ and $X = \blind N R$:
$$
\infer[\blindsym_{L2}]
{\Gamma, \sign {\blind N R} Y \vdash M}
{
\deduce{\Gamma, \sign {\blind N R} Y \vdash R}{\Pi_1}
&
\deduce{\Gamma, \sign {\blind N R} Y, \sign N Y, R \vdash M}{\Pi_2}
}
$$
Applying the inner induction hypothesis (on proof height) 
to $\Pi_1$ and $\Pi_2$ we obtain two derivations $\Pi_1'$ and $\Pi_2'$ of
$$
\Gamma, \blind N R, Y \vdash R
\quad \hbox{ and }
\quad
\Gamma, \blind N R, Y, \sign N Y, R \vdash M.
$$ 
Next we apply the outer induction hypothesis (on the size of $f(X,Y)$) to 
decompose $\sign N Y$ in the latter sequent to get a derivation $\Pi_2''$ of
$$
\Gamma, \blind N R, N, Y, R \vdash M.
$$
The derivation $\Pi'$ is constructed as follows:
$$
\infer[\blindsym_{L1}]
{\Gamma, \blind N R, Y \vdash M}
{
\deduce{\Gamma, \blind N R, Y \vdash R}{\Pi_1'}
&
\deduce{\Gamma, \blind N R, N, Y, R \vdash M}{\Pi_2''}
}
$$
\item Suppose $\Pi$ ends with $id$. The only non-trivial case is 
when $f(X,Y)$ is used in the rule, that is, we have
$$
M \approx C[f(X,Y)^n, M_1, \ldots, M_k]
$$
where $M_1, \ldots, M_k \in \Gamma$, $C[\ldots]$ is an $E$-context and 
$f(X,Y)$ fills $n$-holes in $C[\ldots].$ We distinguish several cases:
\begin{itemize}
\item There is a guarded subterm $A$ in $M$ or some $M_i$ such that $f(X,Y) \equiv A.$
Note that in this case $A$ must be of the form $f(X',Y')$ for some $X' \equiv X$
and $Y' \equiv Y.$
In this case, $\Pi'$ is constructed as follows:
$$
\infer[gs]
{\Gamma, X,Y \vdash M}
{
  \deduce{\Gamma, X,Y \vdash f(X',Y')}{\Xi}
  &
  \infer[id]
  {\Gamma, X,Y, f(X',Y') \vdash M}{}
}
$$
where $\Xi$ is a derivation formed using $id$ and the right rules for the constructor
$f$.

\item Suppose that there is no subterm $A$ of $M$, $M_1, \ldots, M_k$ such that $A \equiv f(X,Y).$ 
Note that since $M$ is in normal form, we have
$$
C[f(X,Y)^n,M_1,\ldots,M_k] \to^* M. 
$$
and both $C[f(X,Y)^n,M_1,\ldots,M_k]$ and $M$ are quasi $E$-terms. 
Let $x = v(f(X,Y))$. 
It follows from Proposition~\ref{prop:var-abs} that
$$
F_E(C[f(X,Y)^n, M_1,\ldots, M_k]) = C[x^n, F_E(M_1), \ldots, F_E(M_k)] \to^* F_E(M).
$$
Since no subterms of $M$ and $M_1,\ldots,M_k$
are equivalent to $f(X,Y)$, $x$ does not appear in
any of $F_E(M)$, $F_E(M_1), \ldots, F_E(M_k)$.
Now let $a$ be a name that does not occur in $\Gamma$, $X$, $Y$ or $M$.
Since rewriting is invariant under variable/name substitution, by substituting
$a$ for $x$ in the above sequence of rewrites, we have
$$
F_E(C[a^n, M_1,\ldots, M_k]) = C[a^n, F_E(M_1), \ldots, F_E(M_k)] \to^* F_E(M).
$$
Now by Proposition~\ref{prop:var-abs2}, we have
$$
C[a^n, M_1, \ldots, M_k] \to^* M.
$$
By substituting $X$ for $a$ in this sequence, we have
$$
C[X^n,M_1,\ldots,M_k] \longrightarrow_{\Rscr}^* M. 
$$
Thus, in this case, $\Pi'$ is constructed by an application of $id.$

\end{itemize}
\end{itemize}
\qed
\end{proof}

\paragraph{\bf Lemma \ref{lm:decomp2}}
Let $X_1,\ldots, X_k$ be normal terms and let 
$\Pi$ be a cut-free derivation of $\Gamma, \norm{f(X_1,\ldots,X_k)} \vdash M$, where $f \in \Sigma_E.$
Then there exists a cut-free derivation $\Pi'$ of $\Gamma, X_1,\ldots,X_k \vdash M.$ 
\begin{proof}
By induction on $|\Pi|.$ The case where $\Pi$ ends with $id$, or rules in which
$\norm{f(X_1,\ldots,X_k)}$ is not principal, is trivial. The other cases, where
$\Pi$ ends with a rule applied to $\norm{f(X_1,\ldots,X_k)}$, are given in the following.
\begin{itemize}
\item Suppose $\Pi$ ends with $p_L$ on $\norm{f(X_1,\ldots,X_k)}.$ This means
that $\norm{f(X_1,\ldots,X_k)}$ is a guarded term, i.e., it is a pair $\spr U V$ for
some $U$ and $V$, and therefore
$$
f(X_1,\ldots,X_k) \to^* \spr U V.
$$
Let $x = F_E(\spr U V)$. By Proposition~\ref{prop:var-abs}, we have
$$
f(F_E(X_1,),\ldots, F_E(X_k)) \to^* x.
$$
Obviously, $x$ has to occur in $F_E(X_i)$ for some $X_i$. Without loss
of generality we assume that $i = 1.$ This means that there exists
a subterm $A$ of $X_1$ such that $A = \spr {U'} {V'}$ and
$U \equiv U'$ and $V \equiv V'$. That is, $X_1 = C[\spr {U'} {V'}]$
for some context $C[.].$

Let $\Gamma'$ be the set $\Gamma \cup \{X_1,\ldots,X_k\}$. 
Then $\Pi'$ is the derivation
$$
\infer[gs]
{\Gamma, C[\spr {U'} {V'}], X_2, \ldots, X_k \vdash M}
{
  \infer[id]
  {\Gamma' \vdash \spr {U'} {V'}}{}
  &
  \deduce{\Gamma', \spr {U'} {V'} \vdash M}{\Pi_1}
}
$$
The instance of $id$ above is valid since $\spr U V \approx f(X_1,\ldots,X_k)$
and $X_1,$ $\ldots,$ $X_k \in \Gamma.$ The derivation $\Pi_1$ is obtained by weakening 
$\Pi$ with $X_1,\ldots, X_k$  and applying Lemma~\ref{lm:equiv} to replace
$\spr U V$ with its equivalent $\spr {X'} {Y'}.$
The cases where $\norm{f(X_1,\ldots,X_k)}$ is headed with some other constructor
are proved analogously.

\item Suppose $\Pi$ ends with $gs$ on $\norm{f(X_1,\ldots,X_k)}$:
$$
\infer[gs]
{\Gamma, \norm{f(X_1,\ldots,X_k)} \vdash M}
{
\deduce{\Gamma' \vdash A}{\Pi_1}
&
\deduce{\Gamma', A \vdash M}{\Pi_2}
}
$$
where $A$ is a guarded subterm of $\norm{f(X_1,\ldots,X_k)}$
and $$\Gamma' = \Gamma \cup \{\norm{f(X_1,\ldots,X_k)}\}.$$ 
Using a similar argument as in the previous case (utilising 
Proposition~\ref{prop:var-abs}), we can show that
$A \equiv A'$ for some $A'$ which is either an $E$-alien 
subterm of some $X_i$ (w.l.o.g., assume $i=1$) or 
a guarded subterm of an $E$-alien subterm of $X_i$.
In either case, we have that $X_1 = C[A']$ for some context $C[.].$
Then $\Pi'$ is
$$
\infer[gs]
{\Gamma, X_1, \ldots, X_k \vdash M}
{
  \deduce{\Gamma'' \vdash A'}{\Pi_1'}
  &
  \deduce{\Gamma'',A' \vdash M}{\Pi_2'}
}
$$
where $\Gamma'' = \Gamma \cup \{X_1,\ldots,X_k\}$ and $\Pi_1'$ and $\Pi_2'$ are obtained
by applying the induction hypothesis on $\Pi_1$ and $\Pi_2$, followed by
applications of Lemma~\ref{lm:equiv} to replace $A$ with its equivalent $A'.$
\end{itemize}
\qed
\end{proof}

\paragraph{\bf Lemma \ref{lm:decomp3}}
Let $M_1, \ldots, M_k$ be terms in normal form and let $C[\ldots]$ be a $k$-hole
$E$-context. If $\Gamma, \norm{C[M_1,\ldots,M_k]} \vdash M$ is cut-free derivable,
then $\Gamma, M_1, \ldots, M_k \vdash M$ is also cut-free derivable.
\begin{proof}
By induction on the size of $C[\ldots]$, Lemma~\ref{lm:equiv} and Lemma~\ref{lm:decomp2}. \qed
\end{proof}

\paragraph{\bf Theorem \ref{thm cut elim}}
The cut rule is admissible for $\Sscr$.
\begin{proof}
We give a set of transformation rules for derivations ending with cuts and show
that given any derivation, there is a sequence of reductions that applies to
this derivation, and terminates with a cut free derivation with the same end sequent. 
This is proved by induction on 
the height of the {\em left premise derivation} immediately above the cut rule. 
This measure is called the {\em cut rank}. As usual in cut elimination, we proceed by eliminating
the topmost instances of cut with the highest rank. So in the following, we suppose a given
derivation $\Pi$ ending with a cut rule, which is the only cut in $\Pi$, and then
show how to transform this to a cut free derivation $\Pi'.$

The cut reduction is driven by the left premise derivation of the cut. 
We distinguish several cases, based on the last rule of the left premise derivation. 
\begin{enumerate}
\item Suppose the left premise of $\Pi$ ends with either $p_R$, $e_R$, $\signsym_R$
or $\blindsym_R$,  thus $\Pi$ is 
$$
\infer[cut]
{\Gamma \vdash R}
{
 \infer[\rho]
 {\Gamma \vdash f(M, N)}
 {\deduce{\Gamma \vdash M}{\Pi_1} & \deduce{\Gamma \vdash N}{\Pi_2}}
 &
 \deduce{\Gamma, f(M, N) \vdash R}{\Pi_3}
}
$$
where $f$ is a constructor and $\rho$ is its right introduction rule.
By Lemma~\ref{lm:decomp1}, 
we have a cut free derivation $\Pi_3'$ of 
$\Gamma, M, N \vdash R.$ By applying Lemma~\ref{lm:weak} to $\Pi_2$, we also have a cut-free derivation
$\Pi_2'$ of $\Gamma, M \vdash N$ such that $|\Pi_2| = |\Pi_2'|.$
The above cut is then reduced to
$$
\infer[cut]
{\Gamma \vdash R}
{
 \deduce{\Gamma \vdash M}{\Pi_1}
 &
 \infer[cut]
 {\Gamma, M \vdash R}
 {
  \deduce{\Gamma, M \vdash N} {\Pi_2'}
  &
  \deduce{\Gamma, M, N \vdash R}{\Pi_3'}
 }
}
\enspace .
$$
These two cuts can then be eliminated by induction hypothesis since their
left premises are of smaller height than the left premise of $\Pi.$

\item Suppose the left premise of the cut ends with a left rule 
acting on $\Gamma.$ 
We show here the case where the left-rule has only one premise; generalisation 
to the other case (with two premises) is straightforward. Therefore $\Pi$ is of the form:
$$
\infer[cut]
{\Gamma \vdash R}
{
 \infer[\rho]
 {\Gamma \vdash M}
 {
  \deduce{\Gamma' \vdash M}{\Pi_1}
 }
 &
 \deduce{\Gamma, M \vdash R}{\Pi_2}
}
$$
By inspection of the inference rules in Figure~\ref{fig:msg}, it is clear that 
in the rule $\rho$ above, we have $\Gamma \subseteq \Gamma'$. We can therefore weaken $\Pi_2$
to a derivation $\Pi_2'$ of $\Gamma', M \vdash R$ with $|\Pi_2| = |\Pi_2'|$. 
The cut is then reduced as follows.
$$
\infer[\rho]
{\Gamma \vdash R}
{
 \infer[cut]
 {\Gamma' \vdash R}
 {
   \deduce{\Gamma' \vdash M}{\Pi_1}
   &
   \deduce{\Gamma', M \vdash R}{\Pi_2}
 }
}
$$
The cut rule above $\rho$ can be eliminated by induction hypothesis, 
the height of the left premise of the cut is smaller than the left
premise of the original cut. 

\item Suppose the left premise of the cut ends with $gs$, but using a subterm from
the right hand side of the sequent, i.e., $\Pi$ is
$$
\infer[cut]
{\Gamma \vdash R}
{
  \infer[gs]
  {\Gamma \vdash C[A]}
  {
    \deduce{\Gamma \vdash A}{\Pi_1}
    &
    \deduce{\Gamma, A \vdash C[A]}{\Pi_2}
  }
  &
  \deduce{\Gamma, C[A] \vdash R}{\Pi_3}
}
$$
If $C[.]$ is an empty context, then $C[A] \equiv A$ and the above cut reduces to
$$
\infer[cut]
{\Gamma \vdash R}
{
  \deduce{\Gamma \vdash A}{\Pi_1}
  &
  \deduce{\Gamma, A \vdash R}{\Pi_3}
}
$$
This cut can be reduced by induction hypothesis, since the height of the left premise
derivation ($\Pi_1$) is smaller than the left premise of the original cut.
If $C[.]$ is a non-empty context, the above cut reduces to the following two cuts:
$$
\infer[cut]
{\Gamma \vdash R}
{
  \deduce{\Gamma \vdash A}{\Pi_1}
  &
  \infer[cut]
  {\Gamma, A \vdash R}
  {
    \deduce{\Gamma, A \vdash C[A]}{\Pi_2}
    &
    \deduce{\Gamma, A, C[A] \vdash R}{\Pi_3'}
  } 
}
$$ 
The derivation $\Pi_3'$ is obtained by weakening $\Pi_3$ with $A$ (Lemma~\ref{lm:weak}).
Both cuts can be removed by induction hypothesis (the upper cut followed by the lower cut).

\item Suppose the left premise of the cut ends with the $id$-rule:
$$
\infer[cut]
{\Gamma \vdash R}
{
 \infer[id]
 {\Gamma \vdash M}
 {}
 &
 \deduce{\Gamma, M \vdash R}{\Pi_1}
}
$$
where $M = \norm{C[M_1,\ldots,M_k]}$ and 
$M_1, \ldots, M_k \in \Gamma.$ 
In this case, we apply Lemma~\ref{lm:decomp3} to $\Pi_1$, hence we get
a cut free derivation $\Pi'$ of 
$\Gamma \vdash R.$
\end{enumerate}
\qed
\end{proof}

\subsection{Proofs for Section~\ref{sec:normal}}

\begin{lemma}
\label{lm:perm1}
Let $\Pi$ be a cut-free derivation of $\Gamma \vdash M.$ Then there is a cut-free 
derivation of the same sequent such that all the right rules
appear above left rules.
\end{lemma}
\begin{proof}
We permute any offending right rules up over any left
rules. This is done by induction on the number of occurrences of the offending rules.
We first show the case where $\Pi$ has at most one offending right rule.
In this case, we show, by induction on the height of $\Pi$, that any offending
right-introduction rule can be permuted up in the derivation tree until it 
is above any left-introduction rule. 
We show here a non-trivial case involving $gs$; the others are treated analogously.
Suppose $\Pi$ is as shown below at left
where $\rho$ denotes a right introduction rule for the constructor $f$
and $A$ is a guarded subterm of $M$.
By the weakening lemma (Lemma~\ref{lm:weak}), we have a derivation
$\Pi_3'$ of $\Gamma, A \vdash N$ with $|\Pi_3'| = |\Pi_3|.$ The
original derivation $\Pi$ is
then transformed into the derivation shown below at right:
$$
\infer[\rho]
{\Gamma \vdash f(M,N)}
{
 \infer[gs]
 {\Gamma \vdash M}
 {
  \deduce{\Gamma \vdash A}{\Pi_1}
  &
  \deduce{\Gamma, A \vdash M}{\Pi_2}
 }
 &
 \deduce{\Gamma \vdash N}{\Pi_3}
}
\hspace{1cm}
\infer[gs]
{\Gamma \vdash f(M,N)}
{
 \deduce{\Gamma \vdash A}{\Pi_1}
 &
 \infer[\rho]
 {\Gamma, A \vdash f(M,N)}
 {
  \deduce{\Gamma,A \vdash M}{\Pi_2}
  &
  \deduce{\Gamma, A \vdash N}{\Pi_3'}
 }
}
$$ 

The rule $\rho$ in the right premise can then be further permuted up (i.e., if
$\Pi_2$ or $\Pi_3'$ ends with a left rule) by induction hypothesis. 

The derivation $\Pi'$ is then constructed by repeatedly applying the above
transformation to the topmost offending rules until all of them appear
above left-introduction rules.
\qed
\end{proof}

\paragraph{\bf Proposition \ref{prop:norm}}
If $\Gamma \vdash M$ is derivable then it has a normal derivation. 
\begin{proof}
Let $\Pi$ be a cut-free derivation of $\Gamma \vdash M$. By Lemma~\ref{lm:perm1},
we can assume without loss of generality that all the right rules in
$\Pi$ appear above the left rules. 
We construct a normal derivation $\Pi'$ of the same sequent by induction
on the number of offending left rules in $\Pi$.

We first consider the case where $\Pi$ has at most one offending left rule. 
Let $\Xi$ be a subtree of $\Pi$ where the offending rule occurs, i.e., 
$\Xi$ ends with a branching left rule, whose left premise derivation ends with
a left rule. We show by induction on the height of the left premise derivation of 
the last rule in $\Xi$ that $\Xi$ can be transformed into a normal derivation.
There are two cases to consider: one in which the left premise derivation ends with
a branching left rule and the other where it ends with a non-branching left rule.
We consider the former case here, the latter can be dealt with analogously.
So suppose $\Xi$ is of the form:
$$
\infer[L_1]
{\Gamma_1 \vdash M'}
{
  \infer[L_2]
  {\Gamma_1 \vdash N_1}
  {
    \deduce{\Gamma_1 \vdash N_2}{\Pi_1}
    &
    \deduce{\Gamma_2 \vdash N_1}{\Pi_2} 
  }
  &
  \deduce{\Gamma_3 \vdash M'}{\Pi_3}
}
$$
where $L_1$ is a left rule, and $\Pi_1$, $\Pi_2$ and $\Pi_3$ are normal derivations, 
$\Gamma_2 \supseteq \Gamma_1$ and $\Gamma_3 \supseteq \Gamma_1.$ We first weaken $\Pi_3$ into a derivation 
$\Pi_3'$ of $\Gamma_4 \vdash M'$, where $\Gamma_4 = \Gamma_2 \cup \Gamma_3$. 
Such a weakening can be easily shown to not 
affect the shape of the derivations (i.e., it does not introduce or
remove any rules in $\Pi_3$). $\Xi$ is then transformed into 
$$
\infer[L_2]
{\Gamma_1 \vdash M'}
{
  \deduce{\Gamma_1 \vdash N_2}{\Pi_1}
  &
  \infer[L_1]
  {\Gamma_2 \vdash M'}
  {
    \deduce{\Gamma_2 \vdash N_1}{\Pi_2}
    &
    \deduce{\Gamma_4 \vdash M'}{\Pi_3'}
  }
}
$$
By inspection of the rules in Figure~\ref{fig:msg}, it can be shown
that this transformation is valid for any pair of left rules $(L_1,L_2).$
Note that this transformation may introduce at most two offending
left rules, i.e., if $\Pi_1$ and/or $\Pi_2$ end with left rules. 
But notice that the left premise derivations of both $L_1$ and $L_2$
in this case have smaller height than the left premise derivation 
of $L_1$ in $\Xi$. By induction hypothesis, the right premise
derivation of $L_2$ can be transformed into a normal derivation, say
$\Pi_4$, resulting in 
$$
\infer[L_2]
{\Gamma_1 \vdash M'}
{
  \deduce{\Gamma_1 \vdash N_2}{\Pi_1}
  &
  \deduce{\Gamma_2 \vdash M'}{\Pi_4}
}
$$
By another application of the induction hypothesis, this derivation
can be transformed into a normal derivation.  

The general case where $\Pi$ has more than one offending rules can be
dealt with by transforming the topmost occurrences of the left rule
one by one following the above transformation. 
\qed 
\end{proof}

\paragraph{\bf Proposition \ref{prop:normal}}
Every sequent $\Gamma \vdash M$ is provable in $\Sscr$ if and only if it
is provable in $\Lscr.$
\begin{proof}
This follows immediately from cut elimination for $\Sscr$ and the
normal form for $\Sscr$ (Proposition~\ref{prop:norm}). \qed
\end{proof}

\paragraph{\bf Lemma \ref{lm:St-inv}}
Let $\Pi$ be an $\Lscr$-derivation of $\Gamma \vdash M.$
Then for every sequent $\Gamma' \vdash M'$ occuring in $\Pi$,
we have $\Gamma' \cup \{M'\} \subseteq St(\Gamma \cup \{M\}).$
\begin{proof}
By induction on $|\Pi|.$
It is enough to show that for each rule $\rho$ in $\Lscr$ other than $r$
$$
\infer[\rho]
{\Gamma \vdash M}
{\Gamma' \vdash M'}
$$
we have that $St(\Gamma \cup \{M\}) = St(\Gamma \cup \{M'\})$. 

The non-trivial case is the rule $\blindsym_2$:
$$
\infer[\blindsym_2]
{\Gamma_1, \sign {\blind N R} K \vdash M}
{\Gamma_1, \sign {\blind N R} K, \sign N K, R \vdash M}
$$
where $\Gamma = \Gamma_1 \cup \{ \sign {\blind N R} K \}.$
The premise of the rule has a term $\sign N K$ which may not occur in the
conclusion. However, the proper subterms of $\sign N K$ are
included in the proper subterms of $\sign {\blind N R} K$, hence
both the premise and the conclusion have the same set of proper
subterms. 
Notice that $\sign N K \in sst(\Gamma)$, since
both $N$ and $K$ are in $pst(\Gamma).$ 
Therefore in this case we also have that 
$
St(\Gamma \cup \{ M \}) = St(\Gamma' \cup \{M'\}). 
$
\qed
\end{proof}

\paragraph{\bf Lemma \ref{lm:linear}}
If there is an $\Lscr$-derivation of $\Gamma \vdash M$ then there is an $\Lscr$-derivation
of the same sequent whose length is at most quadratic with respect to the size
of $\Gamma \cup \{M\}.$
\begin{proof}
  We first note that any derivation of $\Gamma \vdash M$ can be turned
  into one in which every sequent in the derivation occurs exactly
  once on a branch. Our rules preserve their principal formula when
  read upwards from conclusion to premise, hence the left hand sides
  of the sequents as we go up a branch accumulate more and more
  formulae. That is, they form an increasing chain. At worst, each
  such rule adds only one formula from $St(\Gamma \cup \{M\})$.  Thus,
  by Lemma~\ref{lm:St-inv}, the number of different sequents on a
  branch is bounded by the cardinality of $St(\Gamma \cup \{M\})$,
  which is quadratic in the size of $\Gamma \cup \{M\}.$ \qed
\end{proof}

\begin{lemma}
\label{lm:right-deducibility}
The decidability of the relation $\Vdash_\Rscr$ is polynomially
reducible to the decidability of elementary deduction $\Vdash_E$.
\end{lemma}
\begin{proof}
  Recall that the relation $\Gamma \Vdash_\Rscr M$ holds if we can
  derive $\Gamma \vdash M$ using only right-rules and $id$.  Here is a
  simple proof search procedure for $\Gamma \vdash M$, using only
  right-rules:
\begin{enumerate}
\item If $\Gamma \vdash M$ is elementarily deducible, then we are
  done.
\item Otherwise, apply a right-introduction rule (backwards) to
  $\Gamma \vdash M$ and repeat step 1 for each obtained premise, and
  so on. If no such rules are applicable, then $\Gamma \vdash M$ is
  not derivable.
\end{enumerate}
There are at most $n$ iterations where $n$ is the number of distinct
subterms of $M.$ Note that the check for elementary deducibility in
step 1 is done on problems of size less or equal to $n$.  \qed
\end{proof}

Before we proceed with proving Theorem~\ref{thm:deducibility}, let us first
define the notion of a {\em principal term} in a left-rule in the proof system $\Lscr$
(we refer to Figure~\ref{fig:linear} in the following definition):
\begin{itemize}
\item $\spr M N$ is the principal term of $lp$.
\item $\enc M K$ is the principal term of $le$.
\item $\sign M K$ is the principal term of $\signsym$.
\item $\blind M K$ is the principal term of $\blindsym_1$.
\item $\sign {\blind M R} K$ is the principal term of $\blindsym_2.$
\item $A$ is the principal term of $ls$.
\end{itemize}

Given a sequent $\Gamma \vdash M$ and a pair of principal-term and a left-rule $(N, \rho)$, 
we say that the pair $(N,\rho)$ is {\em applicable} to the sequent if 
\begin{itemize}
\item $\rho$ is $ls$, $N$ is a guarded subterm of $\Gamma \cup \{M\}$, and
there is an instance of $\rho$ with $\Gamma, N \vdash M$ as its premise;

\item $\rho$ is not $ls$, $N \in \Gamma$,  and there is an instance of $\rho$
with $\Gamma \vdash M$ as its conclusion;
\end{itemize}

Let us assume that the complexity of $\Vdash_E$ is $O(f(n)).$
We note the following two facts:
Given a sequent $\Gamma \vdash M$ and a pair of principal-term and a left-rule $(N, \rho)$, 
\begin{description}
\item[F1] the complexity of checking whether $(N, \rho)$ is applicable to $\Gamma \vdash M$
is $O(n^l f(n))$ for some constant $l$;

\item[F2] if $(N, \rho)$ is applicable to $\Gamma \vdash M$, then there is a unique
sequent $\Gamma' \vdash M$ such that
the sequent below 
is a valid instance of $\rho$:
$$
\infer[\rho]
{\Gamma \vdash M}
{\Gamma' \vdash M}
$$
\end{description}
Note that for (F1) to hold, we need to assume a DAG representation
of sequents with maximal sharing of subterms.
The complexity of checking whether a rule is applicable or not then consists of
\begin{itemize}
\item pointer comparisons;
\item pattern match a subgraph with a rule;
\item checking equality module associativity and commutativity (for the rule $\signsym$);
\item and checking $\Vdash_\Rscr$.
\end{itemize}
The first three can be done in polynomial time; and the last one is polynomially
reducible to $\Vdash_E$ (Lemma~\ref{lm:right-deducibility}). 

\paragraph{\bf Theorem \ref{thm:deducibility}}
The decidability of the relation $\Vdash_\Lscr$ is polynomially reducible
to the decidability of elementary deduction $\Vdash_E.$
\begin{proof}
Let $n$ be the size of $St(\Gamma \cup \{M\})$.
Notice that the left-rules in Figure~\ref{fig:linear} are invertible (they accumulate
terms, reading the rules bottom-up), so one does not lose provability by
applying any of the rules in proof search. Thus by blindly applying the left-rules, we
eventually reach a point where the right-rule ($r$) is applicable, hence the original sequent
is provable, or we reach a ``fix point''
where we encounter all previous sequents. For the latter, we show that there is a polynomial
bound to the number of rule applications we need to try before concluding that the original
sequent is not provable. 

Let $M_1, \ldots, M_n$ be an enumeration of the set $St(\Gamma \cup \{M\}).$
Suppose $\Gamma \vdash M$ is provable in $\Lscr.$ Then there is a shortest proof 
in $\Gamma$ where each sequent appears exactly once in the proof.
This also means that there exists a sequence of principal-term-and-rule pairs
$$
(M_{i_1}, \rho_1), \ldots, (M_{i_q}, \rho_q)
$$
that is applicable, successively, to $\Gamma \vdash M$.
Note that $q \leq n$ by Lemma~\ref{lm:linear}.

A simple proof search strategy for $\Gamma \vdash M$ is therefore to 
repeatedly try all possible applicable pairs $(M', \rho')$ for each possible
$M' \in St(\Gamma \cup \{M\})$ and each left-rule $\rho'$. More precisely:
Let $j := 0$ and initialise $\Delta := \Gamma$ 
\begin{enumerate}
\item $j := j + 1$.
\item If $\Delta \Vdash_\Rscr M$ then we are done.
\item Otherwise, for $k = 1$ to $n$ do
  \begin{itemize}
  \item[] for every left-rule $\rho$ do
    \begin{itemize}
  \item[] if $(M_k, \rho)$ is applicable to $\Delta \vdash M$, then
    let $\Gamma_1 \vdash M$ be the unique premise of $\rho$ determined
    by $(M_k, \rho)$ via {\bf F2} and let $\Delta := \Gamma_1$.
  \end{itemize}
    \end{itemize}
\item If $j \leq n$ then go to step 1.
\end{enumerate}
If the original sequent is provable, then at each iteration $j$, 
the algorithm (i.e., step 3) will find the correct
pair $(M_{i_j}, \rho_j)$. (Strictly speaking, the algorithm finds 
the $j$-th pair of a shortest proof, and not necessarily the one given above, since there can be
more than one proof of a given length.) 
If no proof is found after $n$ iterations, then the original sequent is not
provable, since the length of any shortest proof is bound by $n$ by Lemma~\ref{lm:linear}.
By Lemma~\ref{lm:right-deducibility}, step 2 takes $O(n^a f(n))$ for some constant $a$. 
By ({\bf F1}) above, each iteration in step 3 takes $O(n^b f(n))$ for some constant $b$.
Since there are at most $6n$ distinct principal-term-and-rule pairs,
this means step 3 takes $O(6n^{b+1}f(n)).$ Therefore the whole procedure takes
$O(n^{c+1} f(n))$ where $c$ is the greater of $a$ and $b+1$.
Hence the complexity of $\Vdash_\Lscr$ is polynomially reducible to $\Vdash_E.$
\qed
\end{proof}

\subsection{Proofs for Section~\ref{sec:comb}}

The following lemma is similar to Lemma~\ref{lm:equiv}, except that $\equiv$
now denotes equality modulo AC for $\oplus_1, \ldots, \oplus_n.$

\begin{lemma}
\label{lm:equiv2}
Let $\Pi$ be a derivation of $M_1,\ldots, M_k \vdash N$. Then for any
$M_1',$ $\ldots,$ $M_k'$ and $N'$
such that $M_i \equiv M_i'$ and $N \equiv N'$, there is a derivation $\Pi'$
of $M_1',$ $\ldots,$ $M_k' \vdash N'$ such that $|\Pi| = |\Pi'|.$
\end{lemma}

\begin{lemma}
\label{lm:D-decomp1}
Let $X$ and $Y$ be normal terms. 
If $\Gamma, f(X,Y) \vdash M$ is cut-free provable in $\Dscr$, where $f$ is a constructor, 
then $\Gamma, X, Y \vdash M$ is also cut-free provable in $\Dscr$.
\end{lemma}
\begin{proof}
Analogous to the proof of Lemma~\ref{lm:decomp1}.\qed
\end{proof}

\begin{lemma}
\label{lm:D-decomp2}
Let $X_1,$ $\ldots,$ $X_k$ be normal terms and let 
$\Pi$ be a cut-free $\Dscr$-derivation of $\Gamma, \norm{f(X_1,\ldots,X_k)} \vdash M$, where $f \in \Sigma_{E_i}.$
Then there exists a cut-free $\Dscr$-derivation $\Pi'$ of $\Gamma, X_1,\ldots,X_k \vdash M.$ 
\end{lemma}
\begin{proof}
By induction on $|\Pi|$. Most cases are similar to the proof of Lemma~\ref{lm:decomp2}. In particular,
the case involving cross-theory subterms are a straightforward generalisation of those involving 
guarded subterms in the proof of Lemma~\ref{lm:decomp2}.
 
Let $N \equiv \norm{f(X_1,\ldots,X_k)}.$ The new case we need to consider is when $\Pi$ ends with $cs:$
$$
\infer[cs]
{\Gamma, N \vdash M}
{
 \deduce{\Gamma, N \vdash R}{\Pi_1}
 &
 \deduce{\Gamma, N, R \vdash M}{\Pi_2}
}
$$
where $R$ is a cross-theory subterm of $N.$ 

Observe that since $X_1,\ldots,X_k$ are in normal
form, the term $f(X_1,\ldots,X_k)$ is a quasi $E_i$-term. As in the proof of Lemma~\ref{lm:decomp2},
using the variable abstraction 
technique (Proposition~\ref{prop:var-abs} and Proposition~\ref{prop:var-abs2}), we can
show that there must be a cross-theory subterm $R'$ in some $X_i$ (w.l.o.g., assume $i = 1$)
such that $R \equiv R'$. 
Thus $\Pi'$ is constructed straightforwardly by induction hypothesis
on $\Pi_1$ and $\Pi_2$ followed by (possibly) an application of $cs$ on $X_l$
and Lemma~\ref{lm:equiv2}.
\qed
\end{proof}

\begin{lemma}
\label{lm:D-decomp3}
Let $M_1,$ $\ldots,$ $M_k$ be normal terms and let $C[\ldots]$ be a $k$-hole
$E_i$-context. If $\Gamma, \norm{C[M_1,\ldots,M_k]} \vdash M$ is cut-free derivable in $\Dscr$,
then $\Gamma, M_1, \ldots, M_k \vdash M$ is also cut-free derivable in $\Dscr$.
\end{lemma}
\begin{proof}
By induction on the size of $C[\ldots]$, Lemma~\ref{lm:equiv2} and Lemma~\ref{lm:D-decomp2}. \qed
\end{proof}

\paragraph{\bf Theorem \ref{thm cut elim for D}}
The cut rule in $\Dscr$ is admissible. 
\begin{proof}
Analogous to the proof of Theorem~\ref{thm cut elim}, making use of Lemma~\ref{lm:D-decomp1}
and Lemma~\ref{lm:D-decomp3}. \qed
\end{proof}

\end{document}